\documentclass[12pt,oneside, italian]{article}

\usepackage[english]{babel}

\usepackage{cancel}

\usepackage{cite}

\usepackage{latexsym}

\usepackage{amssymb,amsthm,amsmath}

\usepackage{longtable}
\usepackage{graphicx,color}

\newtheorem{theorem}{Theorem}[section]

\newtheorem{proposition}[theorem]{Proposition}

\newtheorem{lemma}[theorem]{Lemma}

\newtheorem{corollary}[theorem]{Corollary}

\newtheorem{remark}[theorem]{Remark}

{\theoremstyle{definition}

\newtheorem*{definition*}{Definition}

\newtheorem*{proposition*}{Proposition}

\newtheorem*{corollary*}{Corollary}

\newtheorem*{lemma*}{Lemma}

\newtheorem*{remark*}{Remark}

\newcommand{\cX}{\mathcal X}

\newcommand{\fq}{\mathbb {F}_q}

\newcommand{\tcSq}{\tilde{\mathcal{S}}_q}
\newcommand{\f}{\mathbb{F}}

\title{AG codes and AG quantum codes from cyclic extensions of the Suzuki and Ree curves}
\author{M. Montanucci, M. Timpanella and G. Zini}

\date{}

\begin{document}

\maketitle 

\begin{abstract}
We investigate several types of linear codes constructed from two families $\tilde{\mathcal S}_q$ and $\tilde{\mathcal R}_q$ of maximal curves over finite fields recently constructed by Skabelund as cyclic covers of the Suzuki and Ree curves. Plane models for such curves are provided, and the Weierstrass semigroup $H(P)$ at an $\mathbb{F}_{q}$-rational point $P$ is shown to be symmetric.
\end{abstract}


\footnote{
This research was partially supported by Ministry for Education, University
and Research of Italy (MIUR) (Project PRIN 2012 ``Geometrie di Galois e
strutture di incidenza'' - Prot. N. 2012XZE22K$_-$005)
 and by the Italian National Group for Algebraic and Geometric Structures
and their Applications (GNSAGA - INdAM).

URL: Maria Montanucci (maria.montanucci@unibas.it), Marco Timpanella (PMtimpanella23@hotmail.it), Giovanni Zini (gzini@math.unifi.it).
}

\section{Introduction}

In \cite{Goppa1,Goppa2} Goppa described a way to use algebraic curves to construct linear error correcting codes, the so called algebraic geometric codes (AG codes).
The construction of an AG code whose alphabet is a finite field $\mathbb F_q$ requires that the underlying curve is $\mathbb F_q$-rational and involves two $\mathbb F_q$-rational divisors $D$ and $G$ on the curve.

In general, to construct a \textit{good} AG code over $\mathbb F_q$ a curve $\mathcal X$ with low genus $g$ with respect to its number of $\mathbb F_q$-rational points is required. In fact, from the Goppa bounds on the parameters of the code, it follows that the relative Singleton defect is upper bounded by the ratio $g/N$, where $N$ can be as large as the number of $\mathbb F_q$-rational points of $\mathcal X$ not in the support of $G$.
Maximal curves $\cX$ over $\mathbb F_q$ attain the Hasse-Weil upper bound for the number of $\mathbb F_q$-rational points with respect to their genus,
$$|\cX(\mathbb{F}_q)| \leq q+1+2g\sqrt{q},$$
and for this reason they have been used in a number of works. Examples of such curves are the Hermitian curve, the GK curve \cite{GK2009}, the GGS curve \cite{GGS}, the Suzuki curve \cite{DL1976}, the Klein quartic when $\sqrt{q}\equiv6\pmod7$ \cite{MEAGHER2008}, together with their quotient curves. Maximal curves often have large automorphism groups which in many cases can be inherited by the code: this can bring good performances in encoding \cite{Joyner2005} and decoding \cite{HLS1995}.

Good bounds on the parameters of one-point codes, that is AG codes arising from divisors $G$ of type $nP$ for a point $P$ of the curve, have been obtained by investigating the Weierstrass semigroup at $P$. These results have been later generalized to codes and semigroups at two or more points; see e.g. \cite{MATTHEWS2001,HOMMA1996,HK2001,CT2005,CK2009,LC2006,Kim1994}.

AG codes from the Hermitian curve have been widely investigated; see \cite{HK2006,HK2005,Tiersma1987,DK2011,HK2006_2,YK1991,Stichtenoth1988} and the references therein. Other constructions based on the  Suzuki curve and the curve with equation $y^q + y = x^{q^r+1}$ can be found in \cite{Matthews2004} and \cite{ST2014}.
More recently, AG Codes from the GK curve have been constructed in \cite{FG2010,CT2016,BMZ}.

In this work we investigate several construction of linear codes starting from certain maximal curves $\tilde{\mathcal R}_q$ and $\tcSq$ which were recently constructed by Skabelund \cite{Skabelund} as cyclic covers of the Ree and Suzuki curves.
In particular, we construct codes of the following types:
\begin{itemize}
\item Multi point AG codes with many automorphisms (Section \ref{manyaut});
\item Dual codes of one-point AG codes and the Feng-Rao minimum distance (Section \ref{codiciduali});
\item AG quantum codes (Section \ref{SectionCSS});
\item weak Castle codes (Section \ref{castle});
\item AG convolutional codes (Section \ref{convo});
\end{itemize}
To this aim, some geometrical information on $\tilde{\mathcal R}_q$ and $\tcSq$ is needed, which is collected in Section \ref{prelSkabe}.
After recalling some known facts from the literature (\!\!\cite{Skabelund,GMQZ}), we also prove that for $\tilde{\mathcal R}_q$ and $\tcSq$ the Weierstrass semigroup at the unique infinite place is symmetric, and we deduce that the $\mathbb{F}_q$-rational places are Weierstrass points.
We also provide new plane models both for $\tilde{\mathcal R}_q$ and $\tcSq$.

\section{Preliminary results}\label{Sec:Preliminaries}

\subsection{Curves and codes}\label{Sec:Preliminaries_Curves}

Let $\mathcal{X}$ be a projective, geometrically irreducible, nonsingular algebraic curve of genus $g$ defined over the finite field $\mathbb{F}_q$ of size $q$.
The symbols $\mathcal{X}(\mathbb{F}_q)$ and $\mathbb{F}_q(\mathcal{X})$ denote the set of $\mathbb{F}_q$-rational points and the field of $\mathbb{F}_q$-rational functions, respectively.
A divisor $D$ on $\mathcal{X}$ is a formal sum $n_1P_1+\cdots+n_rP_r$, where $P_i \in \mathcal{X}(\mathbb{F}_q)$, $n_i \in \mathbb{Z}$, $P_i\neq P_j$ if $i\neq j$.
The divisor $D$ is $\mathbb F_q$-rational if it coincides with its image $n_1P_1^q+\cdots+n_rP_r^q$ under the Frobenius map over $\mathbb F_q$.
For a function $f \in \mathbb{F}_q(\mathcal{X})$, $div(f)$ and $(f)_{\infty}$ indicate the divisor of $f$ and its pole divisor.
Also, the Weierstrass semigroup at $P$ will be indicated by $H(P)$.
The Riemann-Roch space associated with an $\mathbb F_q$-rational divisor $D$ is
$$\mathcal{L}(D) := \{ f \in \mathcal{X}(\mathbb{F}_q) \ : \ div(f)+D \geq 0\}$$
and its dimension over $\mathbb{F}_q$ is denoted by  $\ell(D)$.

Let $P_1,\ldots,P_N\in \mathcal{X}(\mathbb{F}_q)$ be pairwise distinct  points and consider the divisor $D=P_1+\cdots+P_N$ and another $\mathbb F_q$-rational divisor $G$ whose support is disjoint from the support of $D$. The AG code $C(D,G)$ is the image of the linear map $\eta :  \mathcal{L}(G) \to \mathbb{F}_q^N$ given by $\eta(f) = (f(P_1),f(P_2) ,\ldots,f(P_N))$. The code has length $N$ and if $N>\deg(G)$  then $\eta$ is an embedding and the dimension $k$ of $C(D,G)$ is equal to $\ell(G)$. The minimum distance $d$ satisfies $d\geq d^*=N-\deg(G)$, where $d^*$ is called the Goppa designed minimum distance of $C(D,G)$; if in addition $\deg(G)>2g-2$, then by the Riemann-Roch Theorem $k=\deg(G)-g+1$; see \cite[Th. 2.65]{HLP}. The dual code $C^{\bot} (D,G)$ is an $AG$ code with dimension $k^{\bot}=N-k$ and minimum distance $d^{\bot}\geq \deg{G}-2g+2$. If $G=\alpha P$, $\alpha \in \mathbb{N}$, $P \in \mathcal{X}(\mathbb{F}_q)$, the AG codes ${C} (D,G)$ and ${C}^{\bot} (D,G)$ are referred to as one-point AG codes. Let $H(P)$ be the Weierstrass semigroup associated with $P$, that is 
$$H(P) := \{n \in \mathbb{N}_0 \ | \ \exists f \in \mathbb{F}_q(\mathcal{X}), (f)_{\infty}=nP\}= \{\rho_0=0<\rho_1<\rho_2<\cdots\}.$$

Denote by $f_{\ell}\in \mathbb{F}_q(\mathcal{X})$, $\ell\geq 1$, a rational function such that $(f_{\ell})_{\infty}=\rho_{\ell}P$. For $\ell \geq0$, define the \emph{Feng-Rao function} 
$$\nu_\ell := | \{(i,j) \in \mathbb{N}_0^2 \ : \ \rho_i+\rho_j = \rho_{\ell+1}\}|.$$ 
Consider ${C}_{\ell}(P)= {C}^{\bot}(P_1+P_2+\cdots+P_N,\rho_{\ell}P)$, $P,P_1,\ldots,P_N$ pairwise distint points in $\mathcal{X}(\mathbb{F}_q)$. The number 
$$d_{ORD} ({C}_{\ell}(P)) := \min\{\nu_{m} \ : \ m \geq \ell\}$$
is a lower bound for the minimum distance $d({C}_{\ell}(P))$ of the code ${C}_{\ell}(P)$, called the \emph{order bound} or the \emph{Feng-Rao designed minimum distance} of ${C}_{\ell}(P)$; see \cite[Theorem 4.13]{HLP}. Also, by \cite[Theorem 5.24]{HLP}, $d_{ORD} ({C}_{\ell}(P))\geq \ell+1-g$ and equality holds if $\ell \geq 2c-g-1$, where $c=\max \{m \in \mathbb{Z} \ : \ m-1 \notin H(P)\}.$ 

A numerical semigroup is called telescopic if it is generated by a sequence $(a_1,\ldots,a_k)$ such that 
\begin{itemize}
\item $\gcd(a_1, \ldots , a_k)=1$;
\item for each $i=2,\ldots,k$, $a_i/d_i \in \langle a_1/d_{i-1},\ldots, a_{i-1}/d_{i-1}\rangle$, where $d_i=\gcd(a_1,\ldots,a_i)$;
\end{itemize}
see \cite{KP}.
The semigroup $H(P)$ is called symmetric if $2g-1\notin H(P)$. The property of being symmetric for $H(P)$ gives rise to useful simplifications of the computation of $d_{ORD}(C_\ell(P))$, when $\rho_\ell >2g$. The following result is due to Campillo and Farr\'an; see \cite[Theorem 4.6]{CF}.

\begin{proposition} \label{campillo} Let $\cX$ be an algebraic curve of genus $g$ and let $P \in \cX(\mathbb{F}_q)$. If $H(P)$ is a symmetric Weierstrass semigroup then one has $$d_{ORD}(C_\ell(P))=\nu_{\ell},$$ for all $\rho_{\ell+1}=2g-1+e$ with $e \in H(P) \setminus \{0\}$. 
\end{proposition}

\subsection{The automorphism group of an AG code $C(D,G)$}

In the following we use the same notation as in \cite{GK2008,JK2006}.
Let $\mathcal{M}_{N,q}\leq{\rm GL}(N,q)$ be the subgroup of matrices having exactly one non-zero element in each row and column.
For $\gamma\in{\rm Aut}(\fq)$ and $M=(m_{i,j})_{i,j}\in{\rm GL}(N,q)$, let $M^\gamma$ be the matrix $(\gamma(m_{i,j}))_{i,j}$.
Let $\mathcal{W}_{N,q}$ be the semidirect product $\mathcal M_{N,q}\rtimes{\rm Aut}(\fq)$ with multiplication $M_1\gamma_1\cdot M_2\gamma_2:= M_1M_2^\gamma\cdot\gamma_1\gamma_2$.
The \emph{automorphism group} ${\rm Aut}({C}(D,G))$ of ${C}(D,G)$ is the subgroup of $\mathcal{W}_{N,q}$ preserving ${C}(D,G)$, that is,
$$ M\gamma(x_1,\ldots,x_N):=((x_1,\ldots,x_N)\cdot M)^\gamma \in {C}(D,G) \;\;\textrm{for any}\;\; (x_1,\ldots,x_N)\in {C}(D,G). $$
Let ${\rm Aut}_{\fq}(\cX)$ denote the $\fq$-automorphism group of $\cX$. Also, let
$$ {\rm Aut}_{\fq,D,G}(\cX)=\{ \sigma\in{\rm Aut}_{\fq}(\cX)\,\mid\, \sigma(D)=D,\,\sigma(G)\approx_D G \}, $$
where $G'\approx_D G$ if and only if there exists $u\in\fq(\cX)$ such that $G'-G=(u)$ and $u(P_i)=1$ for $i=1,\ldots,N$, and
$$ {\rm Aut}_{\fq,D,G}^+(\cX):=\{ \sigma\in{\rm Aut}_{\fq}(\cX)\,\mid\, \sigma(D)=D,\,\sigma(|G|)=|G| \}, $$
where $|G|=\{G+(f)\mid f\in\overline{\mathbb F}_q(\cX)\}$ is the linear series associated with $G$.
Note that ${\rm Aut}_{\fq,D,G}(\cX)\subseteq {\rm Aut}_{\fq,D,G}^+(\cX)$. 

\begin{remark}\label{Coincidono}
Suppose that ${\rm supp}(D)\cup{\rm supp}(G)=\cX(\fq)$ and each point in ${\rm supp}(G)$ has the same weight in $G$. Then
$$ {\rm Aut}_{\fq,D,G}(\cX) = {\rm Aut}_{\fq,D,G}^+(\cX) = \{\sigma\in{\rm Aut}_{\fq}(\cX)\,\mid\,\sigma({\rm supp}(G))={\rm supp}(G) \}. $$
\end{remark}


\begin{proposition}{\rm(\!\!\cite[Proposition 2.3]{BMZ})}\label{AutSubgroup}
If any non-trivial element of ${\rm Aut}_{\fq}(\cX)$ fixes at most $N-1$ $\fq$-rational points of $\cX$, then ${\rm Aut}({C}(D,G))$ contains a subgroup isomorphic to $ ({\rm Aut}_{\fq,D,G}(\cX)\rtimes{\rm Aut}(\fq))\rtimes \mathbb{F}_q^*$.
\end{proposition}

\subsection{The Suzuki curve ${\mathcal{S}}_q$ and the Ree curve ${\mathcal{R}}_q$} \label{sec1}
For $s\geq1$ and $q=2q_0^2=2^{2s+1}$, the Suzuki curve $\mathcal S_q$ over $\mathbb{F}_q$ is defined by the affine equation $Y^q+Y=X^{q_0}\left(X^q+X\right)$, has genus $q_0(q-1)$ and is maximal over $\mathbb F_{q^4}$. The automorphism group $S(q):={\rm Aut}(\mathcal S_q)$ of $\mathcal S_q$ is isomorphic to the Suzuki group $^2B_2(q)$. 
We state some other properties of $S(q)$; see \cite{GKT2006} for more details.
\begin{itemize}
\item $S(q)$ has size $(q^2+1)q^2(q-1)$ and is a simple group.
\item $S(q)$ is generated by the stabilizer 
$$S(q)_{P_\infty}=\left\{\psi_{a,b,c}:(x,y)\mapsto (ax+b,a^{q_0+1}y+b^{q_0}x+c)\,|\,a,b,c\in\fq ,a\ne0\right \}$$ 
of the unique place centered at the infinite point of $\mathcal{S}_q$, together with the involution $\phi:(x,y)\mapsto(\alpha/\beta,y/\beta)$, where $\alpha:=y^{2q_0}+x^{2q_0+1}$ and $\beta:=xy^{2q_0}+\alpha^{2q_0}$.
\item $S(q)$ has exactly two short orbits on $\mathcal S_q$. One is non-tame of size $q^2+1$, consisting of all $\mathbb F_q$-rational places; the other is tame of size $q^2(q-1)(q+2q_0+1)$, consisting of all $\mathbb F_{q^4}\setminus\mathbb F_q$-rational places. The group $S(q)$ acts $2$-transitively on its non-tame short orbit, and the stabilizer $S(q)_{P,Q}$ of  two distinct $\fq$-rational places $P$ and $Q$ is tame and cyclic.
\end{itemize}

For $s\geq1$ and $q=3q_0^2=3^{2s+1}$, the Ree curve $\mathcal{R}_q$ over $\f_q$ is defined by the affine equations
$$
\mathcal{R}_q:\left\{
\begin{array}{ll}
z^q - z=x^{2q_0}\left( x^q - x \right)\\
y^q - y=x^{q_0}\left( x^q - x \right)
\end{array}
\right.,
$$
has genus $\frac{3}{2}q_0(q-1)(q+q_0+1)$ and is maximal over $\f_{q^6}$.
The automorphism group $R(q):={\rm Aut}(\mathcal{R}_q)$ is isomorphic to the simple Ree group $^2 G_2(q)$.
We state some other properties of $R(q)$; see \cite{Pedersen1992} and \cite{Skabelund}.
\begin{itemize}
\item $R(q)$ has size $(q^3+1)q^3(q-1)$.
\item $R(q)$ is generated by the stabilizer
$$R(q)_{P_\infty} = \left\{\psi_{a,b,c,d}\,|\,a,b,c,d\in\fq,a\ne0\right\},$$
$$ \psi_{a,b,c,d}:(x,y,z)\mapsto(ax+b,a^{q_0+1}y+ab^{q_0}x+c, a^{2q_0+1}z - a^{q_0+1}b^{q_0}y + ab^{2q_0}x + d), $$
of the unique place centered at the infinite point of $\mathcal{R}_q$, together with the involution $\phi:(x,y,z)\mapsto(w_6/w_8,w_{10}/w_8,w_9/w_8)$, for certain polynomial functions $w_i\in\f_3[x,y,z]$.
\item $R(q)$ has exactly two short orbits on $\mathcal{R}_q$.
 One is non-tame of size $q^3+1$, consisting of all $\f_q$-rational places. The other is tame of size $q^3(q-1)(q+1)(q+3q_0+1)$, consisting of all $\mathbb F_{q^6}\setminus\mathbb F_q$-rational places.
\end{itemize}

\subsection{Cyclic extensions of $\mathcal{S}_q$ and $\mathcal{R}_q$} \label{prelSkabe}

\subsubsection{A cyclic extension of $\mathcal{S}_q$}\label{secStilda}

The following constructions are due to D. Skabelund \cite{Skabelund}.
Let $\tilde{\mathcal{S}}_q$ be the curve defined over $\mathbb{F}_q$ by the affine equations 
\begin{equation}\label{equazsq}
\tilde{\mathcal{S}}_q: \begin{cases} y^q+y=x^{q_0}(x^q+x) \\ t^m=x^q+x \end{cases},
\end{equation}
where $m=q-2q_0+1$, $q_0=2^s$ and $q=2q_0^2$, with $s\geq 1$.

The curve $\tcSq$ can be seen as a degree-$m$ Kummer extension $t^m=x^q+x$ of an Artin-Schreier extension $y^q+y=x^{q_0}(x^q+x)$ of the projective line. In particular,

\begin{itemize}
\item[•] $(x)= m \sum\limits_{i=1}^q P_{(0,\alpha_i,0)}- mqP_{\infty}$,
\item[•] $(y)= m(q_0+1)P_{(0,0,0)} + m\sum\limits_{i=2}^q P_{(\alpha_i,0,0)} - m(q_0+q)P_{\infty}$,
\item[•] $(t)= \sum\limits_{i,j=1}^q P_{(\alpha_i,\alpha_j,0)} - q^2P_{\infty}$,
\end{itemize}
where $P_{\infty}$ is the only infinite place of $\tilde{\mathcal{S}}_q$ and $\mathbb{F}_q=\lbrace \alpha_1=0, \alpha_2,...,\alpha_q \rbrace$. Also, $g(\tcSq)=\frac{q^3-2q^2+q}{2}$.
The set of $\mathbb{F}_q$-rational places of $\tilde{\mathcal{S}}_q$ has size $q^2+1$ and consists of the places centered at the affine points of $\tilde{\mathcal{S}}_q$ lying on the plane $t=0$, together with $P_\infty$. These places correspond exactly to the $\mathbb{F}_q$-rational places of $\mathcal S_q$. The automorphism group ${\rm Aut}(\tilde{\mathcal S}_q)$ of $\tcSq$ admits the following subgroups:
\begin{itemize}
\item A cyclic group $C_m$ generated by the automorphism $\tau:(x,y,t)\mapsto(x,y,\lambda t)$, where $\lambda\in\f_{q^4}$ is a primitive $m$-th root of unity; $C_m$ is the Galois group of the cover $\tilde{\mathcal{S}}_q \to \mathcal{S}_q$.
\item A group $LS(q)$ lifted by $S(q)$ and generated by the automorphisms $\tilde{\psi}_{a,b,c}$ ($a,b,c\in\fq$, $a\ne0$) together with an involution $\tilde{\phi}$.  Here, $\tilde{\psi}_{a,b,c}(x,y):=\psi_{a,b,c}(x,y)$ and $\tilde{\psi}_{a,b,c}(t):=\delta t$, where $\delta^m=a$.
Similarly, $\tilde{\phi}(x,y):=\phi(x,y)$, and $\tilde{\phi}(t):=t/\beta$ (see \cite[Section 3]{Skabelund}).
\end{itemize}

The full automorphism group $\rm{Aut}(\tilde{\mathcal{S}}_q)$ of $\tilde{\mathcal{S}}_q$ was computed in \cite{GMQZ} and is a direct product $\tilde{S}(q) \times C_m$, where $\tilde{S}(q)\cong S(q)$. Also, $\rm{Aut}(\tilde{\mathcal{S}}_q)$ has exactly two short orbits: one short orbit $O_1$ has size $q^2+1$ and coincides with $\tilde{\mathcal{S}}_q(\f_q)$; the other short orbit $O_2$ has size $|S(q)|$, and hence the stabilizer in $\rm{Aut}(\tilde{\mathcal{S}}_q) $ of a place in $O_2$ has order $m$. The contribution to the different divisor of every element in $\rm{Aut}(\tilde{\mathcal{S}}_q) $ is also described, as summarized in the following lemma.

\begin{lemma}{\rm{ (\!\!\cite[Theorem 28]{GMQZ}}) } \label{ContributionsSuzuki}
Let $\sigma\in \tilde{S}(q)\setminus\{id\}$ and $C_m=\langle \tau\rangle$.
Denote by $o(\sigma)$ the order of $\sigma$.
Then $i(\tau^k)=q^2+1$ for all $k=1,\ldots,m-1$ and one of the following cases occurs.
\begin{itemize}
\item $o(\sigma)=2$, $i(\sigma)=m(2q_0+1)+1$, and $i(\sigma\tau^k)=1$ for all $k=1,\ldots,m-1$;
\item $o(\sigma)=4$, $i(\sigma)=m+1$, and $i(\sigma\tau^k)=1$ for all $k=1,\ldots,m-1$;
\item $o(\sigma)\mid(q-1)$, $i(\sigma)=2$, and $i(\sigma\tau^k)=2$ for all $k=1,\ldots,m-1$;
\item $o(\sigma)\mid(q+2q_0+1)$, $i(\sigma)=0$, and $i(\sigma\tau^k)=0$ for all $k=1,\ldots,m-1$;
\item $o(\sigma)\mid(q-2q_0+1)$, $i(\sigma)=0$, $i(\sigma\tau^j)=4m$ for exactly one $j\in\{1,\ldots,m-1\}$, and $i(\sigma\tau^k)=0$ for all $k\in\{1,\ldots,m-1\}\setminus\{j\}$.
\end{itemize}
\end{lemma}


Denote by $P_{(a,b,c)}$ the unique place centered at the affine point of coordinates $(a,b,c)$ of $\tcSq$ and by $P_\infty$ the place centered at the unique infinite point of $\tcSq$.
Define the functions $z:=y^{2q_0}+x^{2q_0+1}$ and $w:=xy^{2q_0}+z^{2q_0}$, which satisfy $z^q+z=x^{2q_0}(x^q+x)$ and $w^q+w=y^{2q_0}(x^q+x)$. Then
$$(x) =m \sum_{a^q+a=0} P_{(0,a,0)} - (q^2-2qq_0+q) P_\infty,$$
$$(y) = m(q_0+1) P_{(0,0,0)} + m \sum_{b^q+b=0, b \ne 0} P_{(b,0,0)}-(q^2-qq_0+q_0) P_\infty,$$
\begin{equation}\label{divt}
(t)=\sum_{a^q+a=0, b^q+b=0} P_{(a,b,0)} - q^2 P_\infty,
\end{equation}
$$(w)=(q^2+1)(P_{(0,0,0)}-P_\infty),$$
while $v_{P_\infty}(z)=-(q^2-q+2q_0)$ and $v_{P_{(0,0,0)}}(z)=m(2q_0+1)$. In particular,
$$\langle q^2-2qq_0+q, q^2-qq_0+q_0,q^2-q+2q_0,q^2,q^2+1 \rangle \subseteq H(P_\infty),$$
as already noted in \cite{Skabelund}.
We show that $H(P_\infty)$ is symmetric. This property will be used in the next sections to the construction of codes.

\begin{theorem}\label{Ssimmetrico}
The Weierstrass semigroup $H(P_\infty)$ at $P_\infty$ is symmetric.
\end{theorem}

\begin{proof}
Consider the plane curve $\mathcal C$ defined by $t^m=x^q+x$, and the $p$-group $G_1\leq{\rm Aut}(\mathcal C)$ of order $q$ of translations $(x,t)\mapsto(x+a,t)$, $t\in\mathbb{F}_q$.
Then the quotient curve $\mathcal C /G_1$ is rational (\!\!\cite[Lemma 12.1 (iii)(g)]{HKT}) and the pole $\bar{P}_\infty$ of $t$ on $\mathcal C$ is the unique place which ramifies in $\mathcal C\rightarrow \mathcal C/G_1$ (\!\!\cite[Lemma 12.1 (iii)(d)]{HKT}).
Then $H(\bar{P}_\infty)$ is symmetric by \cite[Lemma 62]{KK}; see also \cite[Page 36]{Lew}.

Now consider the curve $\tcSq$, which is an Artin-Schreier extension $y^q+y=x^{q_0}(x^q+x)$ of $\mathcal C$; the Galois group $G_2$ of $\mathcal \tcSq\to\mathcal C$ is a $p$-group of order $q$.
Since $\tcSq$ is an $\mathbb{F}_{q^4}$-maximal curve and $G_2$ is a $p$-group, there exists exactly one place which ramifies in $\tcSq\to\mathcal C$; see \cite[Lemma 11.129 and Section 11]{HKT}. This place is $P_\infty$, as $P_\infty$ is totally ramified in over $\overline{F}_q(x)$.
Therefore, by \cite[Lemma 62]{KK}, $H(P_\infty)$ is symmetric.
\end{proof}

\begin{corollary}\label{SWeierstrassPoints}
The $q^2+1$ $\mathbb{F}_q$-rational places of $\tcSq$ are Weierstrass points.
\end{corollary}

\begin{proof}
By Theorem \ref{Ssimmetrico} and \cite[Proposition 50]{KK}, $P_\infty$ is a Weierstrass point.
Then any place in the same orbit of $P_\infty$ under ${\rm Aut}(\tcSq)$ is a Weierstrass point, and the claim follows because the $q^2+1$ $\mathbb{F}_q$-rational places of $\tcSq$ form a unique orbit.
\end{proof}

We provide a plane model for $\tcSq$ as follows.

\begin{theorem}\label{pianosq}
A plane model of degree $q^2$ for $\tcSq$ is given by the equation $F(y,t)=0$, where
\begin{equation}\label{modpiano}
F(y,t)= y^{q^2}+y^q t^{m(q-1)}+y^q+y t^{m(q-1)}+t^{(q+q_0)m}.
\end{equation}
The coordinate functions $y$ and $t$ of this model are exactly the coordinate functions $y$ and $t$ of the model \eqref{equazsq}.
\end{theorem}
\begin{proof}
By direct computation, \begin{center} $y^{q^2}+y^qt^{m(q-1)}+y^q+yt^{m(q-1)}+t^{(q+q_0)m}=x^{qq_0}(x^{q^2}+x^q)+x^{q_0}(x^q+x)(x^q+x)^{q-1}+(x^q+x)^{q+q_0}=0.$ \end{center}
Then $\tcSq$ has a plane model $F_1(y,t)=0$, where $F_1$ is an absolutely irreducible factor of $F(y,t)$.
Let $F(y,t)=F_1(y,t)\cdot F_2(y,t)$ with $F_2\in\overline{\mathbb{F}}_q[y,t]$, let $Z$ be the homogeneous coordinate in the $yt$-plane, and with abuse of notation denote by $P_\infty$ also the point of $\tcSq$ at which the place $P_\infty$ is centered.
Consider the intersection multiplicity $I_{P_\infty}(F\cap Z)$ of the curve $F(Y,T)=0$ with the plane $Z=0$ at the point $P_\infty$.
From Equation \eqref{modpiano}, $I_{P_\infty}(F\cap Z)=q^2$; on the other hand, $I_{P_\infty}(F\cap Z)=I_{P_\infty}(F_1\cap Z)+I_{P_\infty}(F_2\cap Z)$. Thus, $I_{P_\infty}(F_2\cap Z)=q^2-I_{P_\infty}(F_1\cap Z)$.
From the properties of the valuation $v_{P_\infty}(\cdot)$ at the place $P_\infty$, we have $-q^2=v_P(t)=I_P(F_1\cap T)-I_P(F_1\cap Z)$; since $P_\infty$ does not lie on the plane $T=0$, we have $I_{P_\infty}(F_1\cap T)=0$. Thus, $I_{P_\infty}(F_1\cap Z)=q^2$.
Therefore, $I_{P_\infty}(F_2\cap Z)=0$. As $P_\infty$ is the only intersection point of $F(Y,T)=0$ with the plane at infinity $Z=0$, this implies that $F_2(Y,T)$ is a constant in $\overline{\mathbb F}_q$. Hence, $F(y,t)$ is absolutely irreducible and the claim is proved.
\end{proof}

\subsubsection{A cyclic extension of $\mathcal{R}_q$}\label{secRtilda}

Let $\tilde{\mathcal{R}}_q$ be the curve defined over $\mathbb{F}_q$ by the affine equations
 $$\tilde{\mathcal{R}}_q : \begin{cases} t^m=x^q-x \\ z^q-z=x^{2q_0}(x^q-x)\\y^q-y=x^{q_0}(x^q-x)\end{cases} ,$$
 where $m=q-3q_0+1$ with $q_0=3^s$, $q=3q_0^2=3^{2s+1}$, $s\geq 1$. 
As for the curve $\tcSq$, from the theory of Artin-Schreier and Kummer extensions we obtain the principal divisors of its coordinate functions $x,y,z,t$:
\begin{itemize}
\item[•] $(x)= m\sum\limits_{i,j=1}^q P_{(0,\beta_i,\beta_j,0)}-mq^2P_{\infty}$,
\item[•] $(y)= m(q_0+1)\sum\limits_{i=1}^q P_{(0,0,\beta_i,0)} + \sum\limits_{i=2}^q \sum\limits_{j=1}^q P_{(\beta_i,0,\beta_j,0)} -mq(q_0+q)P_{\infty}$,
\item[•] $(z)= m(2q_0+1)\sum\limits_{i=1}^q P_{(0,\beta_i,0,0)} + m\sum\limits_{i=2}^q \sum\limits_{j=1}^q P_{(\beta_i,\beta_j,0,0)} -mq(2q_0+q)P_{\infty}$,
\item[•] $(t)= \sum\limits_{i,j,l=1}^q P_{(\beta_i,\beta_j,\beta_l,0)} - q^3P_{\infty}$,
\end{itemize}
where $P_{\infty}$ is the only infinite place of $\tilde{\mathcal{R}}_q$ and $\mathbb{F}_q=\{ \beta_1=0, \beta_2,\ldots,\beta_q \}$.
The genus of $\tilde{\mathcal R}_q$ is $\frac{1}{2}\left(q^4-2q^3+q\right)$.
The set of $\f_q$-rational places of $\tilde{\mathcal R}_q$ has size $q^3+1$ and consists of the places centered at the affine points of $\tilde{\mathcal{R}}_q$ lying on the plane $t=0$, together with $P_\infty$. Also, $\tilde{\mathcal R}_q$ has no $\mathbb F_{q^2}$- or $\mathbb F_{q^3}$-rational places which are not $\mathbb F_q$-rational.
The automorphism group ${\rm Aut}(\tilde{\mathcal R}_q)$ of $\tilde{\mathcal R}_q$ has the following subgroups:
\begin{itemize}
\item A cyclic group $C_m$ generated by the automorphism $\tau:(x,y,t)\mapsto(x,y,\lambda t)$, where $\lambda\in\f_{q^4}$ is a primitive $m$-th root of unity; $C_m$ is the Galois group of the cover $\tilde{\mathcal{R}}_q\to\mathcal{R}_q$.
\item A group $LR(q)$ lifted by $R(q)$ and generated by the automorphisms $\tilde{\psi}_{a,b,c,d}$ ($a,b,c,d\in\fq$, $a\ne0$) together with the involution $\tilde{\phi}$. 
Here, $\tilde{\psi}_{a,b,c,d}(x,y,z):=\psi_{a,b,c,d}(x,y,z)$ and $\tilde{\psi}_{a,b,c,d}(t):=\delta t$, where $\delta^m=a$.
Similarly, $\tilde{\phi}(x,y,z):=\phi(x,y,z)$, and $\tilde{\phi}(t):=t/w_8$; see \cite[Section 4]{Skabelund}.
\end{itemize}

The full automorphism group $\rm{Aut}(\tilde{\mathcal{R}}_q)$ of $\tilde{\mathcal{R}}_q$ was computed in \cite{GMQZ} and is a direct product $\tilde{R}(q)\times C_m$, where $\tilde{R}(q)\cong R(q)$. Also, $\rm{Aut}(\tilde{\mathcal{R}}_q)$ has exactly two short orbits: one short orbit $O_1$ has length equal to $q^3+1$ and coincides with $\tilde{\mathcal{R}}_q(\f_q)$; the other short orbit $O_2$ has size $|R(q)|$, and hence the stabilizer in $\rm{Aut}(\tilde{\mathcal{R}}_q) $ of a place in $O_2$ has order $m$. The contribution to the different divisor of every element in $\rm{Aut}(\tilde{\mathcal{R}}_q) $ is also described, as summarized in the following lemma.

\begin{lemma}{\rm (\!\!\cite{GMQZ})} \label{ContributionsRee}
Let $\sigma\in \tilde{R}_q\setminus\{id\}$ and $C_m=\langle\tau\rangle$. Denote by $o(\sigma)$ the order of $\sigma$. Then $i(\tau^k)=q^3+1$ for all $k=1,\ldots,m-1$ and one of the following cases occurs.
\begin{itemize}
\item $o(\sigma)=3$, $\sigma$ is in the center of a Sylow $3$-subgroup, $i(\sigma)=m(q+3q_0+1)+1=q^2-q+2$, and $i(\sigma\tau^k)=1$ for all $k=1,\ldots,m-1$;
\item $o(\sigma)=3$, $\sigma$ is not in the center of any Sylow $3$-subgroup, $i(\sigma)=m(3q_0+1)+1=q^2-q+2-mq$, and $i(\sigma\tau^k)=1$ for all $k=1,\ldots,m-1$;
\item $o(\sigma)=9$, $i(\sigma)=m+1$, and $i(\sigma\tau^k)=1$ for all $k=1,\ldots,m-1$;
\item $o(\sigma)=2$, $i(\sigma)=q+1$, and $i(\sigma\tau^k)=q+1$ for all $k=1,\ldots,m-1$;
\item $o(\sigma)=6$, $i(\sigma)=1$, and $i(\sigma\tau^k)=1$ for all $k=1,\ldots,m-1$;
\item $o(\sigma)\mid(q-1)$, $o(\sigma)\ne2$, $i(\sigma)=2$, and $i(\sigma\tau^k)=2$ for all $k=1,\ldots,m-1$;
\item $o(\sigma)\mid(q+1)$, $o(\sigma)\ne2$, $i(\sigma)=0$, and $i(\sigma\tau^k)=0$ for all $k=1,\ldots,m-1$;
\item $o(\sigma)\mid(q+3q_0+1)$, $i(\sigma)=0$, and $i(\sigma\tau^k)=0$ for all $k=1,\ldots,m-1$;
\item $o(\sigma)\mid(q-3q_0+1)$, $i(\sigma)=0$, $i(\sigma\tau^j)=6m$ for exactly one $j\in\{1,\ldots,m-1\}$, and $i(\sigma\tau^k)=0$ for all $k\in\{1,\ldots,m-1\}\setminus\{j\}$.
\end{itemize}
\end{lemma}

As for $\tcSq$, we prove the simmetricity of $H(P_\infty)$.

\begin{theorem}\label{Rsimmetrico}
The Weierstrass semigroup $H(P_\infty)$ at $P_\infty$ is symmetric.
\end{theorem}

\begin{proof}
Arguing as in the proof of Theorem \ref{Ssimmetrico} and applying \cite[Lemma 62]{KK} several times, the following steps are proved.
If $\mathcal X$ is the curve defined by $t^m=x^q-x$, then the Weierstrass semigroup at the unique infinite place of $\mathcal X$ is symmetric.
Equation $z^q-z=x^{2q_0}(x^q-x)$ defines an Artin-Schreier extension $\mathcal Y$ of $\mathcal X$, and the Weierstrass semigroup at the unique infinite place of $\mathcal Y$ is symmetric.
Equation $y^q-y=x^{q_0}(x^q-x)$ defines the Artin-Schreier extension $\tilde{\mathcal R}_q$ of $\mathcal Y$, and the Weierstrass semigroup at the unique infinite place of $\tilde{\mathcal R}_q$ is symmetric.
\end{proof}

Arguing as for Corollary \ref{SWeierstrassPoints}, Theorem \ref{Rsimmetrico} implies the following result.

\begin{corollary}
The $q^3+1$ $\mathbb{F}_q$-rational places of $\tilde{\mathcal R}_q$ are Weierstrass points.
\end{corollary}

\begin{remark}
An explicit description of the generators of $H(P_\infty)$ seems to be a challenging task, as well as a description of the Weierstrass semigroup $H(Q_\infty)$ at the unique infinite point $Q_\infty$ of the Ree curve $\mathcal{R}_q$. Only partial results on $H(Q_\infty)$ are known; for instance, $H(Q_\infty)$ has $132$ generators for $q=27$, as shown in {\rm \cite[Table 15]{DE}}.
\end{remark}

We provide a plane model for $\tilde{\mathcal{R}}_q$ as follows.
\begin{theorem}
A plane model of degree $q^3$ for $\tilde{\mathcal{R}}_q$ is given by the following equation: 
$$t^{q^3}-y^{q^2}t^{m(q-1)(q+3q_0}-y^{q^2}t^{mq(q-1)}-y^{q^2}+y^q t^{m(q-1)(q+3q_0+1)}$$
$$+y^q t^{m(q-1)(q+3q_0)}+y^q t^{mq(q-1)}-y t^{m(q-1)(q+3q_0+1)}-t^{mq(q+3q_0+1)}=0.$$
\end{theorem}
\begin{proof}
Starting with the plane model of $\mathcal{R}_q$ given in \cite[Section 12.4]{HKT}, the claim follows by arguing as in the proof of Theorem \ref{pianosq}.
\end{proof}



\section{AG codes from $\tilde{\mathcal S}_q$ and $\tilde{\mathcal R}_q$ with many automorphisms} \label{manyaut}

\subsection{AG codes from $\tcSq$}

In this section we use the notation of Section \ref{secStilda}.
Let
\begin{center}
$\mathcal{G} :=\mathcal{X}(\mathbb{F}_q),\quad \mathcal{D}:=\mathcal{X}(\mathbb{F}_{q^4})\setminus\mathcal{G}.$
\end{center}
The set $\mathcal{G}$ is the intersection between $\tcSq$ and the plane $t=0$. Also, $\mathcal{G}$ and $\mathcal{D}$ are respectively the non-tame orbit $O_1$ and the tame orbit $O_2$ of ${\rm Aut}(\tilde{\mathcal S}_q)$.
Fix $r\in\mathbb{N}$ and define the $\mathbb{F}_{q^4}$-rational divisors
\begin{center}
$G:=\sum_{P\in \mathcal{G}} rP$, \quad  $D:= \sum_{P\in \mathcal{D}} P$,
\end{center}
of degree $r(q^2+1)$ and $q^5-q^4+q^3-q^2$, respectively. Let $C$ be the $[n,k,d]_{q^4}$-AG code $C(D,G)$. Then $C$ has designed minimum distance 
\begin{center}
$d^*=n-\deg(G)=q^5-q^4+q^3-q^2-r(q^2+1).$
\end{center}

Next result follows as a corollary of the Riemann-Roch Theorem, see \cite{Sti}.
\begin{proposition}
If $q-2<r<q^3-q^2$, then
$k=\deg(G)+1-g=r(q^2+1)-\frac{q^3-2q^2+q-2}{2}$.
\end{proposition}

\begin{proposition} \label{moneq}
$C$ is monomially equivalent to the one-point code $C(D,r(q^2+1)P_{\infty})$.
\end{proposition}
\begin{proof}
Let $G'=r(q^2+1)P_{\infty}$; then $G=G'+(t^r)$. Thus, the Riemann-Roch space of $G$ is $\mathcal L(G)=\{ f\cdot t^r \mid f\in L(G') \}$. The codeword of $C(D,G')$ associated to $f\cdot t^r$  is 
$$ ((ft^r)(P_1),\ldots,(ft^r)(P_n))=(f(P_1),\ldots,f(P_n))\cdot M,$$
 where $M$  is a diagonal matrix in which the entries are $t(P_1)^r, \ldots, t(P_n)^r \in \mathbb{F}_{q^4}$. This implies that the matrix $M$ defines a monomial isomorphism between $C$ and $C(D,G')$. 
\end{proof}


We determine a group of automorphisms of $C$ inherited from the automorphisms of $\tilde{\mathcal{S}}_q$.

\begin{lemma}\label{N}
Any non-trivial element of ${\rm Aut}(\tilde{\mathcal{S}}_q)$ fixes at most $q^2+1$ $\mathbb{F}_{q^4}$-rational places of $\tilde{\mathcal{S}}_q$.
\end{lemma}
\begin{proof}
Let $g\in{\rm Aut}(\tcSq)\setminus\{id\}$.
If $g$ is a $2$-element, then $g$ fixes exactly one place of $\tcSq$ because $\tcSq$ is $\mathbb{F}_{q^4}$-maximal; see \cite[Lemma 11.129]{HKT}.
If $g$ is not a $2$-element, then the number of fixed places of $g$ equals $i(g)$, which is determined in Lemma \ref{ContributionsSuzuki}.
Then the claim follows from Lemma \ref{ContributionsSuzuki}.
\end{proof}

Since $n=\deg(D)>q^2+1$, Proposition \ref{AutSubgroup} and Lemma \ref{N} yield the following result.

\begin{corollary}
The automorphism group of $C$ admits a subgroup isomorphic to $$({\rm Aut}(\tilde{\mathcal{S}}_q)\rtimes {\rm Aut}(\mathbb{F}_{q^4}))\rtimes \mathbb{F}_{q^4}^*.$$
\end{corollary}

\subsection{AG codes from $\tilde{\mathcal{R}}_q$}

In this section we use the notation of Section \ref{secRtilda}.
Let
\begin{center}
$\mathcal{G} :=\mathcal{X}(\mathbb{F}_q), \quad \mathcal{D}:=\mathcal{X}(\mathbb{F}_{q^6})\setminus\mathcal{G}.$
\end{center}
The set $\mathcal{G}$ is the intersection between $\tilde{\mathcal{R}}_q$ and the plane $t=0$.
Fix $r \in \mathbb{N}$ and consider the $\mathbb{F}_{q^6}$-rational divisors
\begin{center}
$G:=\sum_{P\in \mathcal{G}} rP$, \quad $D:= \sum_{P\in \mathcal{D}} P$,
\end{center}
 of degree $r(q^3+1)$ and $q^7-q^6+q^4-q^3$, respectively. Let $C$ be the $[n,k,d]_{q^6}$-AG code $C(D,G)$. The designed minimum distance of $C$ is
\begin{center}
$d^*=n-\deg(G)=q^7-q^6+q^4-q^3-r(q^3+1)$.
\end{center}

From the Riemann-Roch Theorem we get the following result.

\begin{proposition}
If $q-2<r< q^4-q^3$, then the dimension of $C$ is \begin{center}
$k=r(q^3+1)-\frac{1}{2}(q^4-2q^3+q-2)$
\end{center}
\end{proposition}

\begin{proposition}
$C$ is monomially equivalent to the one-point code $C(D,r(q^3+1)P_{\infty})$.
\end{proposition}
\begin{proof}
Since $G=r(q^3+1)P_{\infty}+(t^r)$, the claim follows as in the proof of Proposition \ref{moneq}.
\end{proof}


\begin{lemma}\label{NR}
Any non-trivial element of ${\rm Aut}(\tilde{\mathcal R}_q)$ fixes at most $q^3+1$ $\mathbb{F}_{q^6}$-rational places of $\tilde{\mathcal R}_q$.
\end{lemma}

\begin{proof}
Arguing as in the proof of Lemma \ref{N}, the claim follows from Lemma \ref{ContributionsRee}.
\end{proof}

Since $n=\deg(D)>q^3+1$, Proposition \ref{AutSubgroup} and Lemma \ref{NR} yield the following result.

\begin{corollary}
The automorphism group of $C$ admits a subgroup isomorphic to
$$({\rm Aut}(\tilde{\mathcal{R}}_q)\rtimes {\rm Aut}(\mathbb{F}_{q^6}))\rtimes \mathbb{F}_{q^6}^*.$$
\end{corollary}

\section{Dual codes of one-point codes from $\tilde{\mathcal{S}}_q$} \label{codiciduali}

In this section we construct dual codes $C_{\ell}(P_\infty)$ of one-point AG codes on the curve $\tilde{\mathcal{S}}_q$, and we compute explicitely the Feng-Rao minimum distance of $C_{\ell}(P_\infty)$.
Denote $$H(P_{\infty}) =\lbrace 0=\rho_1<\rho_2<...\rbrace .$$
For every $\ell\geq 1$, the Feng-Rao function is defined as $$ \nu_{\ell} := \vert \lbrace (i,j)\in \mathbb{N}_0^2 : \rho_i+\rho_j=\rho_{\ell +1} \rbrace \vert,$$
Let $C_{\ell}(P_{\infty})$ be the dual code $$C_{\ell}(P_{\infty}) =C^{\perp}(D,\rho_{\ell}P_{\infty}),$$
where $D=\sum_{P\in\tcSq(\mathbb{F}_{q^4})\setminus\{P_{\infty}\}} P$ 
is a divisor supported at all $\mathbb{F}_{q^4}$-rational places of $\tcSq$ but $P_\infty$.
Then the \textit{Feng-Rao minimum distance}
 $$d_{ORD}(C_{\ell}(P_{\infty})):=\min\lbrace \nu_m : m\geq \ell \rbrace $$ 
is a lower bound for the minimum distance of $C_{\ell}(P_{\infty})$.
The code $C_{\ell}(P_{\infty}))$ has parameters $[n,k,d]_{q^4}$, where $n=q^5-q^4+q^3$, $k=n-\ell$ and $d\geq d_{ORD}(C_{\ell}(P_{\infty}))$.

\begin{proposition} \label{dord1}
For every $\ell \geq 3g-1$, $d_{ORD}(C_{\ell}(P_{\infty}))=\ell+1-g$.
\end{proposition} 
\begin{proof}
In general, $d_{ORD}(C_{\ell}(P_{\infty}))\geq \ell+1-g$ and equality holds for every $\ell \geq 2c-g-1$, where $c=\max \lbrace m\in \mathbb{Z} : m-1\notin H(P_{\infty}) \rbrace $ is the conductor of $H(P_\infty)$. From Theorem \ref{Ssimmetrico}, $c=2g-1$ and the claim follows.
\end{proof}

\begin{proposition}
For every $\rho_{\ell +1}=2g-1+e$, with $e\in H(P_{\infty})\setminus \lbrace 0 \rbrace$, $$d_{ORD}(C_{\ell}(P_{\infty}))=\nu_\ell.$$
\end{proposition}
\begin{proof}
As $H(P_\infty)$ is symmetric by Theorem \ref{Ssimmetrico}, the claim follows from \cite[Theorem 4.6]{CF}.
\end{proof}

The following tables show explicitly the parameters of the codes $C_{\ell}(P_\infty)$ in the case $q=8$. In particular, we present their length $n=29184$, their dimension $k$, their Feng-Rao minimum distance $d_{ORD}$, an upper bound $n+1-k-d_{ORD}$ for their Singleton defect $\delta=n+1-k-d$, and an upper bound $\frac{n+1-k-d_{ORD}}{n}$ for their relative Singleton defect $\Delta=\frac{\delta}{n}$.

\begin{center}
\begin{scriptsize}
\begin{tabular}{|c|c|c|c|c|c||c|c|c|c|c|c|}
\hline
$n$&$k$&$\rho_{\ell}$&$d_{ORD}$&$\delta \leq $&$\Delta \leq$ & $n$&$k$&$\rho_{\ell}$&$d_{ORD}$&$\delta \leq $&$\Delta \leq$\\ \hline
 $29184$ & $29182$ & $40$ & $2$ & $1$ & $0,0000342$  & $29184$ & $29181$ & $50$ & $2$ & $2$ & $0,0000685$ \\ \hline $29184$ & $29180$ & $60$ & $2$ & $3$ & $0.0001028$ & $29184$ & $29179$ & $64$ & $2$ & $4$ & $0.0001370$\\ \hline $29184$ & $29178$ & $65$ & $3$ & $4$ & $0.0001370$ & $29184$ & $29177$ & $80$ & $3$ & $5$ & $0.0001713$ \\ \hline $29184$ & $29176$ & $90$ & $3$ & $6$ & $0.0002055$ & $29184$ & $29175$ & $100$ & $3$ & $7$ & $0.0002398$ \\ \hline  $29184$ & $29174$ & $104$ & $3$ & $8$ & $0.0002741$ & $29184$ & $29173$ & $105$ & $3$ & $9$ & $0.0003083$ \\ \hline $29184$ & $29172$ & $110$ & $3$ & $10$ & $0.0003426$ & $29184$ & $29171$ & $114$ & $3$ & $11$ & $0.0003769$ \\ \hline $29184$ & $29170$ & $115$ & $3$ & $12$ & $0.0004111$ & $29184$ & $29169$ & $120$ & $3$ & $13$ & $0.0004454$ \\ \hline $29184$ & $29168$ & $124$ & $3$ & $14$ & $0.00047971$ & $29184$ & $29167$ & $125$ & $3$ & $15$ & $0.0005139$ \\ \hline $29184$ & $29166$ & $128$ & $4$ & $15$ & $0.0005139$ & $29184$ & $29165$ & $129$ & $4$ & $16$ & $0.0005482$\\ \hline $29184$ & $29164$ & $130$ & $4$ & $17$ & $0.00058251$ & $29184$ & $29163$ & $140$ & $4$ & $18$ & $0.00061677$ \\ \hline $29184$ & $29162$ & $144$ & $4$ & $19$ & $0.00065104$ & $29184$ & $29161$ & $145$ & $4$ & $20$ & $0.00068530$ \\ \hline $29184$ & $29160$ & $150$ & $4$ & $21$ & $0.00071957$ & $29184$ & $29159$ & $154$ & $4$ & $22$ & $0.00075383$ \\ \hline  $29184$ & $29158$ & $155$ & $4$ & $23$ & $0.00078810$ & $29184$ & $29157$ & $160$ & $4$ & $24$ & $0.00082237$ \\ \hline $29184$ & $29156$ & $164$ & $4$ & $25$ & $0.00085663$ & $29184$ & $29155$ & $165$ & $4$ & $26$ & $0.00089090$ \\ \hline $29184$ & $29154$ & $168$ & $4$ & $27$ & $0.00092516$ & $29184$ & $29153$ & $169$ & $4$ & $28$ & $0.00095943$ \\ \hline $29184$ & $29152$ & $170$ & $4$ & $29$ & $0.00099370$ & $29184$ & $29151$ & $174$ & $4$ & $30$ & $0.0010280$ \\ \hline $29184$ & $29150$ & $175$ & $4$ & $31$ & $0.0010622$ & $29184$ & $29149$ & $178$ & $4$ & $32$ & $0.0010965$ \\
\hline $29184$ & $29148$ & $179$ & $4$ & $33$ & $0.0011308$ & $29184$ & $29147$ & $180$ & $4$ & $34$ & $0.0011650$ \\ \hline  $29184$ & $29146$ & $184$ & $4$ & $35$ & $0.0011993$ & $29184$ & $29145$ & $185$ & $4$ & $36$ & $0.0012335$ \\ \hline $29184$ & $29144$ & $188$ & $4$ & $37$ & $0.0012678$  & $29184$ & $29143$ & $189$ & $4$ & $38$ & $0.0013021$ \\ \hline $29184$ & $29142$ & $190$ & $4$ & $39$ & $0.0013364$ & $29184$ & $29141$ & $192$ & $5$ & $39$ & $0.0013364$ \\ \hline $29184$ & $29140$ & $193$ & $5$ & $40$ & $0.0013706$ & $29184$ & $29139$ & $194$ & $5$ & $41$ & $0.0014049$ \\ \hline $29184$ & $29138$ & $195$ & $5$ & $42$ & $0.0014391$ & $29184$ & $29137$ & $200$ & $5$ & $43$ & $0.0014734$ \\ \hline $29184$ & $29136$ & $204$ & $5$ & $44$ & $0.0015077$ & $29184$ & $29135$ & $205$ & $5$ & $45$ & $0.0015419$ \\ \hline $29184$ & $29134$ & $208$ & $5$ & $46$ & $0.0015762$ & $29184$ & $29133$ & $209$ & $5$ & $47$ & $0.0016105$ \\ \hline $29184$ & $29132$ & $210$ & $5$ & $48$ & $0.0016447$& $29184$ & $29131$ & $214$ & $5$ & $49$ & $0.0016790$ \\ \hline $29184$ & $29130$ & $215$ & $5$ & $50$ & $0.0017133$ & $29184$ & $29129$ & $218$ & $5$ & $51$ & $0.0017475$ \\ \hline $29184$ & $29128$ & $219$ & $5$ & $52$ & $0.0017818$ & $29184$ & $29127$ & $220$ & $5$ & $53$ & $0.0018161$ \\ \hline $29184$ & $29126$ & $224$ & $5$ & $54$ & $0.0018503$ & $29184$ & $29125$ & $225$ & $5$ & $55$ & $0.0018846$ \\ \hline $29184$ & $29124$ & $228$ & $5$ & $56$ & $0.0019189$ & $29184$ & $29123$ & $229$ & $5$ & $57$ & $0.0019531$ \\ \hline  $29184$ & $29122$ & $230$ & $5$ & $58$ & $0.0019874$ & $29184$ & $29121$ & $232$ & $5$ & $59$ & $0.0020217$ \\ \hline $29184$ & $29120$ & $233$ & $5$ & $60$ & $0.0020559$ &  $29184$ & $29119$ & $234$ & $5$ & $61$ & $0.0020902$ \\ \hline $29184$ & $29118$ & $235$ & $5$ & $62$ & $0.0021245$ & $29184$ & $29117$ & $238$ & $5$ & $63$ & $0.0021587$ \\ \hline $29184$ & $29116$ & $239$ & $5$ & $64$ & $0.0021930$ & $29184$ & $29115$ & $240$ & $5$ & $65$ & $0.0022272$ \\ \hline $29184$ & $29114$ & $242$ & $5$ & $66$ & $0.0022615$ & $29184$ & $29113$ & $243$ & $5$ & $67$ & $0.0022958$ \\ \hline $29184$ & $29112$ & $244$ & $5$ & $68$ & $0.0023300$ & $29184$ & $29111$ & $245$ & $5$ & $69$ & $0.0023643$ \\ \hline $29184$ & $29110$ & $248$ & $5$ & $70$ & $0.0023986$ & $29184$ & $29109$ & $249$ & $5$ & $71$ & $0.0024329$ \\ \hline $29184$ & $29108$ & $250$ & $5$ & $72$ & $0.0024671$ & $29184$ & $29107$ & $252$ & $5$ & $73$ & $0.0025014$ \\ \hline $29184$ & $29106$ & $253$ & $5$ & $74$ & $0.0025356$ & $29184$ & $29105$ & $254$ & $5$ & $75$ & $0.0025699$ \\ \hline $29184$ & $29104$ & $255$ & $5$ & $76$ & $0.0026042$ & $29184$ & $29103$ & $256$ & $8$ & $74$ & $0.0025356$ \\ \hline $29184$ & $29102$ & $257$ & $10$ & $73$ & $0.0025014$ & $29184$ & $29101$ & $258$ & $10$ & $74$ & $0.0025356$ \\ \hline $29184$ & $29100$ & $259$ & $10$ & $75$ & $0.0025699$ & $29184$ & $29099$ & $260$ & $10$ & $76$ & $0.0026042$ \\ \hline $29184$ & $29098$ & $264$ & $10$ & $77$ & $0.0026384$ & $29184$ & $29097$ & $265$ & $10$ & $78$ & $0.0026727$ \\ \hline  $29184$ & $29096$ & $268$ & $10$ & $79$ & $0.0027070$ & $29184$ & $29095$ & $269$ & $10$ & $80$ & $0.0027412$ \\ \hline $29184$ & $29094$ & $270$ & $10$ & $81$ & $0.0027755$ & $29184$ & $29093$ & $272$ & $10$ & $82$ & $0.0028098$ \\ \hline $29184$ & $29092$ & $273$ & $10$ & $83$ & $0.0028440$ & $29184$ & $29091$ & $274$ & $10$ & $84$ & $0.0028783$ \\ \hline $29184$ & $29090$ & $275$ & $10$ & $85$ & $0.0029126$ & $29184$ & $29089$ & $278$ & $10$ & $86$ & $0.0029468$ \\ \hline $29184$ & $29088$ & $279$ & $10$ & $87$ & $0.0029811$ & $29184$ & $29087$ & $280$ & $10$ & $88$ & $0.0030153$ \\ \hline $29184$ & $29086$ & $282$ & $10$ & $89$ & $0.0030496$& $29184$ & $29085$ & $283$ & $10$ & $90$ & $0.0030839$ \\ \hline $29184$ & $29084$ & $284$ & $10$ & $91$ & $0.0031182$ & $29184$ & $29083$ & $285$ & $10$ & $92$ & $0.0031524$ \\ \hline $29184$ & $29082$ & $288$ & $10$ & $93$ & $0.0031867$ & $29184$ & $29081$ & $289$ & $10$ & $94$ & $0.0032209$ \\ \hline $29184$ & $29080$ & $290$ & $10$ & $95$ & $0.0032552$  & $29184$ & $29079$ & $292$ & $10$ & $96$ & $0.0032895$ \\ \hline $29184$ & $29078$ & $293$ & $10$ & $97$ & $0.0033237$ & $29184$ & $29077$ & $294$ & $10$ & $98$ & $0.0033580$ \\ \hline $29184$ & $29076$ & $295$ & $10$ & $99$ & $0.0033923$ & $29184$ & $29075$ & $296$ & $10$ & $100$ & $0.0034265$ \\ \hline $29184$ & $29074$ & $297$ & $10$ & $101$ & $0.0034608$ & $29184$ & $29073$ & $298$ & $10$ & $102$ & $0.0034951$ \\ \hline $29184$ & $29072$ & $299$ & $10$ & $103$ & $0.0035293$ & $29184$ & $29071$ & $300$ & $10$ & $104$ & $0.0035636$ \\ \hline $29184$ & $29070$ & $302$ & $10$ & $105$ & $0.0035979$ & $29184$ & $29069$ & $303$ & $10$ & $106$ & $0.0036321$ \\ \hline $29184$ & $29068$ & $304$ & $10$ & $107$ & $0.0036664$ & $29184$ & $29067$ & $305$ & $10$ & $108$ & $0.0037006$ \\ \hline $29184$ & $29066$ & $306$ & $10$ & $109$ & $0.0037349$ & $29184$ & $29065$ & $307$ & $10$ & $110$ & $0.0037692$ \\ \hline $29184$ & $29064$ & $308$ & $10$ & $111$ & $0.0038035$& $29184$ & $29063$ & $309$ & $10$ & $112$ & $0.0038377$ \\ \hline $29184$ & $29062$ & $310$ & $10$ & $113$ & $0.0038720$ & $29184$ & $29061$ & $312$ & $10$ & $114$ & $0.0039062$ \\ \hline
 \end{tabular}\end{scriptsize}\end{center}
\begin{center}\begin{scriptsize}\begin{tabular}{|c|c|c|c|c|c||c|c|c|c|c|c|} 
\hline $n$&$k$&$\rho_{\ell}$&$d_{ORD}$&$\delta \leq $&$\Delta \leq$ & $n$&$k$&$\rho_{\ell}$&$d_{ORD}$&$\delta \leq $&$\Delta \leq$
\\ \hline  $29184$ & $29060$ & $313$ & $10$ & $115$ & $0.0039405$ & $29184$ & $29059$ & $314$ & $10$ & $116$ & $0.0039748$ \\ \hline $29184$ & $29058$ & $315$ & $10$ & $117$ & $0.0040091$ & $29184$ & $29057$ & $316$ & $10$ & $118$ & $0.0040433$ \\ \hline $29184$ & $29056$ & $317$ & $10$ & $119$ & $0.0040776$ & $29184$ & $29055$ & $318$ & $10$ & $120$ & $0.0041118$ \\ \hline $29184$ & $29054$ & $319$ & $10$ & $121$ & $0.0041461$ & $29184$ & $29053$ & $320$ & $10$ & $122$ & $0.0041804$ \\ \hline $29184$ & $29052$ & $321$ & $15$ & $118$ & $0.0040433$ & $29184$ & $29051$ & $322$ & $15$ & $119$ & $0.0040776$ \\ \hline $29184$ & $29050$ & $323$ & $15$ & $120$ & $0.0041118$ & $29184$ & $29049$ & $324$ & $15$ & $121$ & $0.0041461$ \\ \hline $29184$ & $29048$ & $325$ & $15$ & $122$ & $0.0041804$ & $29184$ & $29047$ & $328$ & $15$ & $123$ & $0.0042146$ \\ \hline $29184$ & $29046$ & $329$ & $15$ & $124$ & $0.0042489$ & $29184$ & $29045$ & $330$ & $15$ & $125$ & $0.0042832$ \\ \hline $29184$ & $29044$ & $332$ & $15$ & $126$ & $0.0043175$ & $29184$ & $29043$ & $333$ & $15$ & $127$ & $0.0043517$ \\ \hline $29184$ & $29042$ & $334$ & $15$ & $128$ & $0.0043859$ & $29184$ & $29041$ & $335$ & $15$ & $129$ & $0.0044202$ \\ \hline $29184$ & $29040$ & $336$ & $20$ & $125$ & $0.0042832$ & $29184$ & $29039$ & $337$ & $20$ & $126$ & $0.0043175$ \\ \hline $29184$ & $29038$ & $338$ & $20$ & $127$ & $0.0043517$& $29184$ & $29037$ & $339$ & $20$ & $128$ & $0.0043859$ \\ \hline $29184$ & $29036$ & $340$ & $20$ & $129$ & $0.0044202$& $29184$ & $29035$ & $342$ & $20$ & $130$ & $0.0044545$ \\ \hline $29184$ & $29034$ & $343$ & $20$ & $131$ & $0.0044888$ & $29184$ & $29033$ & $344$ & $20$ & $132$ & $0.0045230$ \\ \hline $29184$ & $29032$ & $345$ & $20$ & $133$ & $0.0045573$ & $29184$ & $29031$ & $346$ & $20$ & $134$ & $0.0045916$ \\ \hline $29184$ & $29030$ & $347$ & $20$ & $135$ & $0.0046258$ & $29184$ & $29029$ & $348$ & $20$ & $136$ & $0.0046601$ \\ \hline $29184$ & $29028$ & $349$ & $20$ & $137$ & $0.0046943$ & $29184$ & $29027$ & $350$ & $20$ & $138$ & $0.0047286$ \\ \hline $29184$ & $29026$ & $352$ & $20$ & $139$ & $0.0047629$ & $29184$ & $29025$ & $353$ & $20$ & $140$ & $0.0047972$ \\ \hline $29184$ & $29024$ & $354$ & $20$ & $141$ & $0.0048314$ & $29184$ & $29023$ & $355$ & $20$ & $142$ & $0.0048657$ \\ \hline $29184$ & $29022$ & $356$ & $20$ & $143$ & $0.0048999$ & $29184$ & $29021$ & $357$ & $20$ & $144$ & $0.0049342$ \\ \hline $29184$ & $29020$ & $358$ & $20$ & $145$ & $0.0049685$ & $29184$ & $29019$ & $359$ & $20$ & $146$ & $0.0050027$ \\ \hline $29184$ & $29018$ & $360$ & $20$ & $147$ & $0.0050370$ & $29184$ & $29017$ & $361$ & $20$ & $148$ & $0.0050713$ \\ \hline $29184$ & $29016$ & $362$ & $20$ & $149$ & $0.0051056$ & $29184$ & $29015$ & $363$ & $20$ & $150$ & $0.0051398$ \\ \hline $29184$ & $29014$ & $364$ & $20$ & $151$ & $0.0051740$ & $29184$ & $29013$ & $365$ & $20$ & $152$ & $0.0052083$ \\ \hline $29184$ & $29012$ & $366$ & $20$ & $153$ & $0.0052426$ & $29184$ & $29011$ & $367$ & $20$ & $154$ & $0.0052769$ \\ \hline $29184$ & $29010$ & $368$ & $20$ & $155$ & $0.0053111$ & $29184$ & $29009$ & $369$ & $20$ & $156$ & $0.0053454$ \\ \hline $29184$ & $29008$ & $370$ & $20$ & $157$ & $0.0053797$ & $29184$ & $29007$ & $371$ & $20$ & $158$ & $0.0054139$ \\ \hline $29184$ & $29006$ & $372$ & $20$ & $159$ & $0.0054482$ & $29184$ & $29005$ & $373$ & $20$ & $160$ & $0.0054824$ \\ \hline $29184$ & $29004$ & $374$ & $20$ & $161$ & $0.0055167$ & $29184$ & $29003$ & $375$ & $20$ & $162$ & $0.0055510$ \\ \hline $29184$ & $29002$ & $376$ & $20$ & $163$ & $0.0055853$ & $29184$ & $29001$ & $377$ & $20$ & $164$ & $0.0056195$ \\ \hline $29184$ & $29000$ & $378$ & $20$ & $165$ & $0.0056538$ & $29184$ & $28999$ & $379$ & $20$ & $166$ & $0.0056881$ \\ \hline $29184$ & $28998$ & $380$ & $20$ & $167$ & $0.0057223$ & $29184$ & $28997$ & $381$ & $30$ & $158$ & $0.0054139$ \\ \hline $29184$ & $28996$ & $382$ & $30$ & $159$ & $0.0054482$ & $29184$ & $28995$ & $383$ & $30$ & $160$ & $0.0054824$ \\ \hline $29184$ & $28994$ & $384$ & $30$ & $161$ & $0.0055167$ & $29184$ & $28993$ & $385$ & $30$ & $162$ & $0.0055510$ \\ \hline $29184$ & $28992$ & $386$ & $30$ & $163$ & $0.0055853$ & $29184$ & $28991$ & $387$ & $30$ & $164$ & $0.0056195$ \\ \hline $29184$ & $28990$ & $388$ & $30$ & $165$ & $0.0056538$ & $29184$ & $28989$ & $389$ & $30$ & $166$ & $0.0056881$ \\ \hline $29184$ & $28988$ & $390$ & $30$ & $167$ & $0.0057223$ & $29184$ & $28987$ & $392$ & $30$ & $168$ & $0.0057566$ \\ \hline $29184$ & $28986$ & $393$ & $30$ & $169$ & $0.0057908$ & $29184$ & $28985$ & $394$ & $30$ & $170$ & $0.0058251$ \\ \hline $29184$ & $28984$ & $395$ & $30$ & $171$ & $0.0058594$ & $29184$ & $28983$ & $396$ & $30$ & $172$ & $0.0058936$ \\ \hline $29184$ & $28982$ & $397$ & $30$ & $173$ & $0.0059279$ & $29184$ & $28981$ & $398$ & $30$ & $174$ & $0.0059622$ \\ \hline $29184$ & $28980$ & $399$ & $30$ & $175$ & $0.0059965$ & $29184$ & $28979$ & $400$ & $30$ & $176$ & $0.0060307$ \\ \hline $29184$ & $28978$ & $401$ & $40$ & $167$ & $0.0057223$ & $29184$ & $28977$ & $402$ & $40$ & $168$ & $0.0057566$ \\ \hline
\end{tabular}\end{scriptsize}\end{center}

\begin{center}\begin{scriptsize}\begin{tabular}{|c|c|c|c|c|c||c|c|c|c|c|c|} 
\hline $n$&$k$&$\rho_{\ell}$&$d_{ORD}$&$\delta \leq $&$\Delta \leq$ & $n$&$k$&$\rho_{\ell}$&$d_{ORD}$&$\delta \leq $&$\Delta \leq$ \\ \hline 
$29184$ & $28976$ & $403$ & $40$ & $169$ & $0.0057908$ & $29184$ & $28975$ & $404$ & $40$ & $170$ & $0.0058251$ \\ \hline $29184$ & $28974$ & $405$ & $40$ & $171$ & $0.0058594$ & $29184$ & $28973$ & $406$ & $40$ & $172$ & $0.0058936$ \\ \hline $29184$ & $28972$ & $407$ & $40$ & $173$ & $0.0059279$ & $29184$ & $28971$ & $408$ & $40$ & $174$ & $0.0059622$ \\ \hline $29184$ & $28970$ & $409$ & $40$ & $175$ & $0.0059965$ & $29184$ & $28969$ & $410$ & $40$ & $176$ & $0.0060307$ \\ \hline $29184$ & $28968$ & $411$ & $40$ & $177$ & $0.0060650$ & $29184$ & $28967$ & $412$ & $40$ & $178$ & $0.0060992$ \\ \hline $29184$ & $28966$ & $413$ & $40$ & $179$ & $0.0061335$ & $29184$ & $28965$ & $414$ & $40$ & $180$ & $0.0061678$ \\ \hline $29184$ & $28964$ & $415$ & $40$ & $181$ & $0.0062020$ & $29184$ & $28963$ & $416$ & $40$ & $182$ & $0.0062363$ \\ \hline $29184$ & $28962$ & $417$ & $40$ & $183$ & $0.0062706$ & $29184$ & $28961$ & $418$ & $40$ & $184$ & $0.0063048$ \\ \hline $29184$ & $28960$ & $419$ & $40$ & $185$ & $0.0063391$ & $29184$ & $28959$ & $420$ & $40$ & $186$ & $0.0063733$ \\ \hline $29184$ & $28958$ & $421$ & $40$ & $187$ & $0.0064076$ & $29184$ & $28957$ & $422$ & $40$ & $188$ & $0.0064419$ \\ \hline $29184$ & $28956$ & $423$ & $40$ & $189$ & $0.0064762$ & $29184$ & $28955$ & $424$ & $40$ & $190$ & $0.0065104$ \\ \hline $29184$ & $28954$ & $425$ & $40$ & $191$ & $0.0065447$ & $29184$ & $28953$ & $426$ & $40$ & $192$ & $0.0065789$ \\ \hline $29184$ & $28952$ & $427$ & $40$ & $193$ & $0.0066132$ & $29184$ & $28951$ & $428$ & $40$ & $194$ & $0.0066475$ \\ \hline $29184$ & $28950$ & $429$ & $40$ & $195$ & $0.0066817$ & $29184$ & $28949$ & $430$ & $40$ & $196$ & $0.0067160$ \\ \hline $29184$ & $28948$ & $431$ & $50$ & $187$ & $0.0064076$ & $29184$ & $28947$ & $432$ & $50$ & $188$ & $0.0064419$ \\ \hline $29184$ & $28946$ & $433$ & $50$ & $189$ & $0.0064762$ & $29184$ & $28945$ & $434$ & $50$ & $190$ & $0.0065104$ \\ \hline $29184$ & $28944$ & $435$ & $50$ & $191$ & $0.0065447$ & $29184$ & $28943$ & $436$ & $50$ & $192$ & $0.0065789$ \\ \hline $29184$ & $28942$ & $437$ & $50$ & $193$ & $0.0066132$ & $29184$ & $28941$ & $438$ & $50$ & $194$ & $0.0066475$ \\ \hline $29184$ & $28940$ & $439$ & $50$ & $195$ & $0.0066817$ & $29184$ & $28939$ & $440$ & $50$ & $196$ & $0.0067160$ \\ \hline $29184$ & $28938$ & $441$ & $60$ & $187$ & $0.0064076$ & $29184$ & $28937$ & $442$ & $60$ & $188$ & $0.0064419$ \\ \hline $29184$ & $28936$ & $443$ & $60$ & $189$ & $0.0064762$ & $29184$ & $28935$ & $444$ & $60$ & $190$ & $0.0065104$ \\ \hline $29184$ & $28934$ & $445$ & $60$ & $191$ & $0.0065447$ & $29184$ & $28933$ & $446$ & $60$ & $192$ & $0.0065789$ \\ \hline $29184$ & $28932$ & $447$ & $60$ & $193$ & $0.0066132$ & $29184$ & $28931$ & $448$ & $60$ & $194$ & $0.0066475$ \\ \hline $29184$ & $28930$ & $449$ & $60$ & $195$ & $0.0066817$ & $29184$ & $28929$ & $450$ & $60$ & $196$ & $0.0067160$ \\ \hline $29184$ & $28928$ & $451$ & $64$ & $193$ & $0.0066132$ & $29184$ & $28927$ & $452$ & $64$ & $194$ & $0.0066475$ \\ \hline $29184$ & $28926$ & $453$ & $64$ & $195$ & $0.0066817$ & $29184$ & $28925$ & $454$ & $64$ & $196$ & $0.0067160$ \\ \hline $29184$ & $28924$ & $455$ & $65$ & $196$ & $0.0067160$& $29184$ & $28923$ & $456$ & $80$ & $182$ & $0.0062363$ \\ \hline $29184$ & $28922$ & $457$ & $80$ & $183$ & $0.0062706$ & $29184$ & $28921$ & $458$ & $80$ & $184$ & $0.0063048$ \\ \hline $29184$ & $28920$ & $459$ & $80$ & $185$ & $0.0063391$& $29184$ & $28919$ & $460$ & $80$ & $186$ & $0.0063733$ \\ \hline $29184$ & $28918$ & $461$ & $80$ & $187$ & $0.0064076$ & $29184$ & $28917$ & $462$ & $80$ & $188$ & $0.0064419$ \\ \hline $29184$ & $28916$ & $463$ & $80$ & $189$ & $0.0064762$ & $29184$ & $28915$ & $464$ & $80$ & $190$ & $0.0065104$ \\ \hline $29184$ & $28914$ & $465$ & $80$ & $191$ & $0.0065447$ & $29184$ & $28913$ & $466$ & $80$ & $192$ & $0.0065789$ \\ \hline $29184$ & $28912$ & $467$ & $80$ & $193$ & $0.0066132$ & $29184$ & $28911$ & $468$ & $80$ & $194$ & $0.0066475$ \\ \hline $29184$ & $28910$ & $469$ & $80$ & $195$ & $0.0066817$ & $29184$ & $28909$ & $470$ & $80$ & $196$ & $0.0067160$ \\ \hline $29184$ & $28908$ & $471$ & $90$ & $187$ & $0.0064076$ & $29184$ & $28907$ & $472$ & $90$ & $188$ & $0.0064419$ \\ \hline $29184$ & $28906$ & $473$ & $90$ & $189$ & $0.0064762$ & $29184$ & $28905$ & $474$ & $90$ & $190$ & $0.0065104$ \\ \hline $29184$ & $28904$ & $475$ & $90$ & $191$ & $0.0065447$ & $29184$ & $28903$ & $476$ & $90$ & $192$ & $0.0065789$ \\ \hline $29184$ & $28902$ & $477$ & $90$ & $193$ & $0.0066132$ & $29184$ & $28901$ & $478$ & $90$ & $194$ & $0.0066475$ \\ \hline $29184$ & $28900$ & $479$ & $90$ & $195$ & $0.0066817$ & $29184$ & $28899$ & $480$ & $90$ & $196$ & $0.0067160$ \\ \hline $29184$ & $28898$ & $481$ & $100$ & $187$ & $0.0064076$ & $29184$ & $28897$ & $482$ & $100$ & $188$ & $0.0064419$ \\ \hline $29184$ & $28896$ & $483$ & $100$ & $189$ & $0.0064762$ & $29184$ & $28895$ & $484$ & $100$ & $190$ & $0.0065104$ \\ \hline $29184$ & $28894$ & $485$ & $100$ & $191$ & $0.0065447$ & $29184$ & $28893$ & $486$ & $100$ & $192$ & $0.0065789$ \\ \hline
\end{tabular}\end{scriptsize}\end{center}

\begin{center}\begin{scriptsize}\begin{tabular}{|c|c|c|c|c|c||c|c|c|c|c|c|} 
\hline $n$&$k$&$\rho_{\ell}$&$d_{ORD}$&$\delta \leq $&$\Delta \leq$ & $n$&$k$&$\rho_{\ell}$&$d_{ORD}$&$\delta \leq $&$\Delta \leq$ \\ \hline 
 $29184$ & $28892$ & $487$ & $100$ & $193$ & $0.0066132$ & $29184$ & $28891$ & $488$ & $100$ & $194$ & $0.0066475$ \\ \hline $29184$ & $28890$ & $489$ & $100$ & $195$ & $0.0066817$ & $29184$ & $28889$ & $490$ & $100$ & $196$ & $0.0067160$ \\ \hline $29184$ & $28888$ & $491$ & $104$ & $193$ & $0.0066132$& $29184$ & $28887$ & $492$ & $104$ & $194$ & $0.0066475$ \\ \hline $29184$ & $28886$ & $493$ & $104$ & $195$ & $0.0066817$ & $29184$ & $28885$ & $494$ & $104$ & $196$ & $0.0067160$ \\ \hline $29184$ & $28884$ & $495$ & $105$ & $196$ & $0.0067160$ & $29184$ & $28883$ & $496$ & $110$ & $192$ & $0.0065789$ \\ \hline $29184$ & $28882$ & $497$ & $110$ & $193$ & $0.0066132$ & $29184$ & $28881$ & $498$ & $110$ & $194$ & $0.0066475$ \\ \hline $29184$ & $28880$ & $499$ & $110$ & $195$ & $0.0066817$ & $29184$ & $28879$ & $500$ & $110$ & $196$ & $0.0067160$ \\ \hline $29184$ & $28878$ & $501$ & $114$ & $193$ & $0.0066132$& $29184$ & $28877$ & $502$ & $114$ & $194$ & $0.0066475$ \\ \hline $29184$ & $28876$ & $503$ & $114$ & $195$ & $0.0066817$ & $29184$ & $28875$ & $504$ & $114$ & $196$ & $0.0067160$ \\ \hline $29184$ & $28874$ & $505$ & $115$ & $196$ & $0.0067160$ & $29184$ & $28873$ & $506$ & $120$ & $192$ & $0.0065789$ \\ \hline $29184$ & $28872$ & $507$ & $120$ & $193$ & $0.0066132$ & $29184$ & $28871$ & $508$ & $120$ & $194$ & $0.0066475$ \\ \hline $29184$ & $28870$ & $509$ & $120$ & $195$ & $0.0066817$ & $29184$ & $28869$ & $510$ & $120$ & $196$ & $0.0067160$ \\ \hline $29184$ & $28868$ & $511$ & $124$ & $193$ & $0.0066132$ & $29184$ & $28867$ & $512$ & $124$ & $194$ & $0.0066475$ \\ \hline $29184$ & $28866$ & $513$ & $124$ & $195$ & $0.0066817$ & $29184$ & $28865$ & $514$ & $124$ & $196$ & $0.0067160$ \\ \hline $29184$ & $28864$ & $515$ & $125$ & $196$ & $0.0067160$ & $29184$ & $28863$ & $516$ & $128$ & $194$ & $0.0066475$ \\ \hline $29184$ & $28862$ & $517$ & $128$ & $195$ & $0.0066817$ & $29184$ & $28861$ & $518$ & $128$ & $196$ & $0.0067160$ \\ \hline $29184$ & $28860$ & $519$ & $129$ & $196$ & $0.0067160$ & $29184$ & $28859$ & $520$ & $130$ & $196$ & $0.0067160$ \\ \hline $29184$ & $28858$ & $521$ & $140$ & $187$ & $0.0064076$ & $29184$ & $28857$ & $522$ & $140$ & $188$ & $0.0064419$ \\ \hline $29184$ & $28856$ & $523$ & $140$ & $189$ & $0.0064762$ & $29184$ & $28855$ & $524$ & $140$ & $190$ & $0.0065104$ \\ \hline $29184$ & $28854$ & $525$ & $140$ & $191$ & $0.0065447$ & $29184$ & $28853$ & $526$ & $140$ & $192$ & $0.0065789$ \\ \hline $29184$ & $28852$ & $527$ & $140$ & $193$ & $0.0066132$ & $29184$ & $28851$ & $528$ & $140$ & $194$ & $0.0066475$ \\ \hline $29184$ & $28850$ & $529$ & $140$ & $195$ & $0.0066817$ & $29184$ & $28849$ & $530$ & $140$ & $196$ & $0.0067160$ \\ \hline $29184$ & $28848$ & $531$ & $144$ & $193$ & $0.0066132$ & $29184$ & $28847$ & $532$ & $144$ & $194$ & $0.0066475$ \\ \hline $29184$ & $28846$ & $533$ & $144$ & $195$ & $0.0066817$ & $29184$ & $28845$ & $534$ & $144$ & $196$ & $0.0067160$ \\ \hline $29184$ & $28844$ & $535$ & $145$ & $196$ & $0.0067160$ & $29184$ & $28843$ & $536$ & $150$ & $192$ & $0.0065789$ \\ \hline $29184$ & $28842$ & $537$ & $150$ & $193$ & $0.0066132$ & $29184$ & $28841$ & $538$ & $150$ & $194$ & $0.0066475$ \\ \hline $29184$ & $28840$ & $539$ & $150$ & $195$ & $0.0066817$ & $29184$ & $28839$ & $540$ & $150$ & $196$ & $0.0067160$ \\ \hline $29184$ & $28838$ & $541$ & $154$ & $193$ & $0.0066132$ & $29184$ & $28837$ & $542$ & $154$ & $194$ & $0.0066475$ \\ \hline $29184$ & $28836$ & $543$ & $154$ & $195$ & $0.0066817$ & $29184$ & $28835$ & $544$ & $154$ & $196$ & $0.0067160$ \\ \hline $29184$ & $28834$ & $545$ & $155$ & $196$ & $0.0067160$ & $29184$ & $28833$ & $546$ & $160$ & $192$ & $0.0065789$ \\ \hline $29184$ & $28832$ & $547$ & $160$ & $193$ & $0.0066132$ & $29184$ & $28831$ & $548$ & $160$ & $194$ & $0.0066475$ \\ \hline $29184$ & $28830$ & $549$ & $160$ & $195$ & $0.0066817$ & $29184$ & $28829$ & $550$ & $160$ & $196$ & $0.0067160$ \\ \hline $29184$ & $28828$ & $551$ & $164$ & $193$ & $0.0066132$ & $29184$ & $28827$ & $552$ & $164$ & $194$ & $0.0066475$ \\ \hline $29184$ & $28826$ & $553$ & $164$ & $195$ & $0.0066817$ & $29184$ & $28825$ & $554$ & $164$ & $196$ & $0.0067160$ \\ \hline $29184$ & $28824$ & $555$ & $165$ & $196$ & $0.0067160$ &$29184$ & $28823$ & $556$ & $168$ & $194$ & $0.0066475$ \\ \hline $29184$ & $28822$ & $557$ & $168$ & $195$ & $0.0066817$& $29184$ & $28821$ & $558$ & $168$ & $196$ & $0.0067160$ \\ \hline $29184$ & $28820$ & $559$ & $169$ & $196$ & $0.0067160$ & $29184$ & $28819$ & $560$ & $170$ & $196$ & $0.0067160$ \\ \hline $29184$ & $28818$ & $561$ & $174$ & $193$ & $0.0066132$ & $29184$ & $28817$ & $562$ & $174$ & $194$ & $0.0066475$ \\ \hline $29184$ & $28816$ & $563$ & $174$ & $195$ & $0.0066817$ & $29184$ & $28815$ & $564$ & $174$ & $196$ & $0.0067160$ \\ \hline $29184$ & $28814$ & $565$ & $175$ & $196$ & $0.0067160$ & $29184$ & $28813$ & $566$ & $178$ & $194$ & $0.0066475$ \\ \hline $29184$ & $28812$ & $567$ & $178$ & $195$ & $0.0066817$& $29184$ & $28811$ & $568$ & $178$ & $196$ & $0.0067160$ \\ \hline $29184$ & $28810$ & $569$ & $179$ & $196$ & $0.0067160$ & $29184$ & $28809$ & $570$ & $180$ & $196$ & $0.0067160$ \\ \hline
\end{tabular}\end{scriptsize}\end{center}

\begin{center}\begin{scriptsize}\begin{tabular}{|c|c|c|c|c|c||c|c|c|c|c|c|} 
\hline $n$&$k$&$\rho_{\ell}$&$d_{ORD}$&$\delta \leq $&$\Delta \leq$ & $n$&$k$&$\rho_{\ell}$&$d_{ORD}$&$\delta \leq $&$\Delta \leq$ \\ \hline 
$29184$ & $28808$ & $571$ & $184$ & $193$ & $0.0066132$ & $29184$ & $28807$ & $572$ & $184$ & $194$ & $0.0066475$ \\ \hline $29184$ & $28806$ & $573$ & $184$ & $195$ & $0.0066817$ & $29184$ & $28805$ & $574$ & $184$ & $196$ & $0.0067160$ \\ \hline $29184$ & $28804$ & $575$ & $185$ & $196$ & $0.0067160$ & $29184$ & $28803$ & $576$ & $188$ & $194$ & $0.0066475$ \\ \hline $29184$ & $28802$ & $577$ & $188$ & $195$ & $0.0066817$ & $29184$ & $28801$ & $578$ & $188$ & $196$ & $0.0067160$ \\ \hline $29184$ & $28800$ & $579$ & $189$ & $196$ & $0.0067160$ & $29184$ & $28799$ & $580$ & $190$ & $196$ & $0.0067160$ \\ \hline $29184$ & $28798$ & $581$ & $192$ & $195$ & $0.0066817$ & $29184$ & $28797$ & $582$ & $192$ & $196$ & $0.0067160$ \\ \hline $29184$ & $28796$ & $583$ & $193$ & $196$ & $0.0067160$ & $29184$ & $28795$ & $584$ & $194$ & $196$ & $0.0067160$ \\ \hline $29184$ & $28794$ & $585$ & $195$ & $196$ & $0.0067160$ & $29184$ & $28793$ & $586$ & $200$ & $192$ & $0.0065789$ \\ \hline $29184$ & $28792$ & $587$ & $200$ & $193$ & $0.0066132$ & $29184$ & $28791$ & $588$ & $200$ & $194$ & $0.0066475$ \\ \hline $29184$ & $28790$ & $589$ & $200$ & $195$ & $0.0066817$ & $29184$ & $28789$ & $590$ & $200$ & $196$ & $0.0067160$ \\ \hline $29184$ & $28788$ & $591$ & $204$ & $193$ & $0.0066132$ & $29184$ & $28787$ & $592$ & $204$ & $194$ & $0.0066475$ \\ \hline $29184$ & $28786$ & $593$ & $204$ & $195$ & $0.0066817$ & $29184$ & $28785$ & $594$ & $204$ & $196$ & $0.0067160$ \\ \hline $29184$ & $28784$ & $595$ & $205$ & $196$ & $0.0067160$ & $29184$ & $28783$ & $596$ & $208$ & $194$ & $0.0066475$ \\ \hline $29184$ & $28782$ & $597$ & $208$ & $195$ & $0.0066817$ & $29184$ & $28781$ & $598$ & $208$ & $196$ & $0.0067160$ \\ \hline $29184$ & $28780$ & $599$ & $209$ & $196$ & $0.0067160$ & $29184$ & $28779$ & $600$ & $210$ & $196$ & $0.0067160$ \\ \hline $29184$ & $28778$ & $601$ & $214$ & $193$ & $0.0066132$ & $29184$ & $28777$ & $602$ & $214$ & $194$ & $0.0066475$ \\ \hline $29184$ & $28776$ & $603$ & $214$ & $195$ & $0.0066817$& $29184$ & $28775$ & $604$ & $214$ & $196$ & $0.0067160$ \\ \hline $29184$ & $28774$ & $605$ & $215$ & $196$ & $0.0067160$ & $29184$ & $28773$ & $606$ & $218$ & $194$ & $0.0066475$ \\ \hline $29184$ & $28772$ & $607$ & $218$ & $195$ & $0.0066817$ & $29184$ & $28771$ & $608$ & $218$ & $196$ & $0.0067160$ \\ \hline $29184$ & $28770$ & $609$ & $219$ & $196$ & $0.0067160$ & $29184$ & $28769$ & $610$ & $220$ & $196$ & $0.0067160$ \\ \hline $29184$ & $28768$ & $611$ & $224$ & $193$ & $0.0066132$ & $29184$ & $28767$ & $612$ & $224$ & $194$ & $0.0066475$ \\ \hline $29184$ & $28766$ & $613$ & $224$ & $195$ & $0.0066817$ & $29184$ & $28765$ & $614$ & $224$ & $196$ & $0.0067160$ \\ \hline $29184$ & $28764$ & $615$ & $225$ & $196$ & $0.0067160$ & $29184$ & $28763$ & $616$ & $228$ & $194$ & $0.0066475$ \\ \hline $29184$ & $28762$ & $617$ & $228$ & $195$ & $0.0066817$ & $29184$ & $28761$ & $618$ & $228$ & $196$ & $0.0067160$ \\ \hline $29184$ & $28760$ & $619$ & $229$ & $196$ & $0.0067160$ & $29184$ & $28759$ & $620$ & $230$ & $196$ & $0.0067160$ \\ \hline $29184$ & $28758$ & $621$ & $232$ & $195$ & $0.0066817$ & $29184$ & $28757$ & $622$ & $232$ & $196$ & $0.0067160$ \\ \hline $29184$ & $28756$ & $623$ & $233$ & $196$ & $0.0067160$ & $29184$ & $28755$ & $624$ & $234$ & $196$ & $0.0067160$ \\ \hline $29184$ & $28754$ & $625$ & $235$ & $196$ & $0.0067160$ & $29184$ & $28753$ & $626$ & $238$ & $194$ & $0.0066475$ \\ \hline $29184$ & $28752$ & $627$ & $238$ & $195$ & $0.0066817$ & $29184$ & $28751$ & $628$ & $238$ & $196$ & $0.0067160$ \\ \hline $29184$ & $28750$ & $629$ & $239$ & $196$ & $0.0067160$& $29184$ & $28749$ & $630$ & $240$ & $196$ & $0.0067160$ \\ \hline $29184$ & $28748$ & $631$ & $242$ & $195$ & $0.0066817$ & $29184$ & $28747$ & $632$ & $242$ & $196$ & $0.0067160$ \\ \hline $29184$ & $28746$ & $633$ & $243$ & $196$ & $0.0067160$ & $29184$ & $28745$ & $634$ & $244$ & $196$ & $0.0067160$ \\ \hline $29184$ & $28744$ & $635$ & $245$ & $196$ & $0.0067160$ & $29184$ & $28743$ & $636$ & $248$ & $194$ & $0.0066475$ \\ \hline $29184$ & $28742$ & $637$ & $248$ & $195$ & $0.0066817$ & $29184$ & $28741$ & $638$ & $248$ & $196$ & $0.0067160$ \\ \hline $29184$ & $28740$ & $639$ & $249$ & $196$ & $0.0067160$ & $29184$ & $28739$ & $640$ & $250$ & $196$ & $0.0067160$ \\ \hline $29184$ & $28738$ & $641$ & $252$ & $195$ & $0.0066817$ & $29184$ & $28737$ & $642$ & $252$ & $196$ & $0.0067160$ \\ \hline $29184$ & $28736$ & $643$ & $253$ & $196$ & $0.0067160$ & $29184$ & $28735$ & $644$ & $254$ & $196$ & $0.0067160$ \\ \hline $29184$ & $28734$ & $645$ & $255$ & $196$ & $0.0067160$ & $29184$ & $28733$ & $646$ & $256$ & $196$ & $0.0067160$ \\ \hline $29184$ & $28732$ & $647$ & $257$ & $196$ & $0.0067160$ & $29184$ & $28731$ & $648$ & $258$ & $196$ & $0.0067160$ \\ \hline $29184$ & $28730$ & $649$ & $259$ & $196$ & $0.0067160$ & $29184$ & $28729$ & $650$ & $260$ & $196$ & $0.0067160$ \\ \hline $29184$ & $28728$ & $651$ & $264$ & $193$ & $0.0066132$ & $29184$ & $28727$ & $652$ & $264$ & $194$ & $0.0066475$ \\ \hline $29184$ & $28726$ & $653$ & $264$ & $195$ & $0.0066817$ & $29184$ & $28725$ & $654$ & $264$ & $196$ & $0.0067160$ \\ \hline 
\end{tabular}\end{scriptsize}\end{center}

\begin{center}\begin{scriptsize}\begin{tabular}{|c|c|c|c|c|c||c|c|c|c|c|c|} 
\hline $n$&$k$&$\rho_{\ell}$&$d_{ORD}$&$\delta \leq $&$\Delta \leq$ & $n$&$k$&$\rho_{\ell}$&$d_{ORD}$&$\delta \leq $&$\Delta \leq$ \\ \hline 
$29184$ & $28724$ & $655$ & $265$ & $196$ & $0.0067160$ & $29184$ & $28723$ & $656$ & $268$ & $194$ & $0.0066475$ \\ \hline $29184$ & $28722$ & $657$ & $268$ & $195$ & $0.0066817$ & $29184$ & $28721$ & $658$ & $268$ & $196$ & $0.0067160$ \\ \hline $29184$ & $28720$ & $659$ & $269$ & $196$ & $0.0067160$& $29184$ & $28719$ & $660$ & $270$ & $196$ & $0.0067160$ \\ \hline $29184$ & $28718$ & $661$ & $272$ & $195$ & $0.0066817$ & $29184$ & $28717$ & $662$ & $272$ & $196$ & $0.0067160$ \\ \hline $29184$ & $28716$ & $663$ & $273$ & $196$ & $0.0067160$ & $29184$ & $28715$ & $664$ & $274$ & $196$ & $0.0067160$ \\ \hline $29184$ & $28714$ & $665$ & $275$ & $196$ & $0.0067160$ & $29184$ & $28713$ & $666$ & $278$ & $194$ & $0.0066475$ \\ \hline $29184$ & $28712$ & $667$ & $278$ & $195$ & $0.0066817$ & $29184$ & $28711$ & $668$ & $278$ & $196$ & $0.0067160$ \\ \hline $29184$ & $28710$ & $669$ & $279$ & $196$ & $0.0067160$ & $29184$ & $28709$ & $670$ & $280$ & $196$ & $0.0067160$ \\ \hline $29184$ & $28708$ & $671$ & $282$ & $195$ & $0.0066817$ & $29184$ & $28707$ & $672$ & $282$ & $196$ & $0.0067160$ \\ \hline $29184$ & $28706$ & $673$ & $283$ & $196$ & $0.0067160$ & $29184$ & $28705$ & $674$ & $284$ & $196$ & $0.0067160$ \\ \hline $29184$ & $28704$ & $675$ & $285$ & $196$ & $0.0067160$ & $29184$ & $28703$ & $676$ & $288$ & $194$ & $0.0066475$ \\ \hline $29184$ & $28702$ & $677$ & $288$ & $195$ & $0.0066817$ & $29184$ & $28701$ & $678$ & $288$ & $196$ & $0.0067160$ \\ \hline $29184$ & $28700$ & $679$ & $289$ & $196$ & $0.0067160$ & $29184$ & $28699$ & $680$ & $290$ & $196$ & $0.0067160$ \\ \hline $29184$ & $28698$ & $681$ & $292$ & $195$ & $0.0066817$ & $29184$ & $28697$ & $682$ & $292$ & $196$ & $0.0067160$ \\ \hline $29184$ & $28696$ & $683$ & $293$ & $196$ & $0.0067160$& $29184$ & $28695$ & $684$ & $294$ & $196$ & $0.0067160$ \\ \hline $29184$ & $28694$ & $685$ & $295$ & $196$ & $0.0067160$ & $29184$ & $28693$ & $686$ & $296$ & $196$ & $0.0067160$ \\ \hline $29184$ & $28692$ & $687$ & $297$ & $196$ & $0.0067160$& $29184$ & $28691$ & $688$ & $298$ & $196$ & $0.0067160$ \\ \hline $29184$ & $28690$ & $689$ & $299$ & $196$ & $0.0067160$ & $29184$ & $28689$ & $690$ & $300$ & $196$ & $0.0067160$ \\ \hline $29184$ & $28688$ & $691$ & $302$ & $195$ & $0.0066817$ & $29184$ & $28687$ & $692$ & $302$ & $196$ & $0.0067160$ \\ \hline $29184$ & $28686$ & $693$ & $303$ & $196$ & $0.0067160$ & $29184$ & $28685$ & $694$ & $304$ & $196$ & $0.0067160$ \\ \hline $29184$ & $28684$ & $695$ & $305$ & $196$ & $0.0067160$ & $29184$ & $28683$ & $696$ & $306$ & $196$ & $0.0067160$ \\ \hline $29184$ & $28682$ & $697$ & $307$ & $196$ & $0.0067160$& $29184$ & $28681$ & $698$ & $308$ & $196$ & $0.0067160$ \\ \hline $29184$ & $28680$ & $699$ & $309$ & $196$ & $0.0067160$ & $29184$ & $28679$ & $700$ & $310$ & $196$ & $0.0067160$ \\ \hline $29184$ & $28678$ & $701$ & $312$ & $195$ & $0.0066817$ & $29184$ & $28677$ & $702$ & $312$ & $196$ & $0.0067160$ \\ \hline $29184$ & $28676$ & $703$ & $313$ & $196$ & $0.0067160$ & $29184$ & $28675$ & $704$ & $314$ & $196$ & $0.0067160$ \\ \hline $29184$ & $28674$ & $705$ & $315$ & $196$ & $0.0067160$ & $29184$ & $28673$ & $706$ & $316$ & $196$ & $0.0067160$ \\ \hline $29184$ & $28672$ & $707$ & $317$ & $196$ & $0.0067160$ & $29184$ & $28671$ & $708$ & $318$ & $196$ & $0.0067160$ \\ \hline $29184$ & $28670$ & $709$ & $319$ & $196$ & $0.0067160$ & $29184$ & $28669$ & $710$ & $320$ & $196$ & $0.0067160$ \\ \hline $29184$ & $28668$ & $711$ & $321$ & $196$ & $0.0067160$& $29184$ & $28667$ & $712$ & $322$ & $196$ & $0.0067160$ \\ \hline $29184$ & $28666$ & $713$ & $323$ & $196$ & $0.0067160$& $29184$ & $28665$ & $714$ & $324$ & $196$ & $0.0067160$ \\ \hline $29184$ & $28664$ & $715$ & $325$ & $196$ & $0.0067160$& $29184$ & $28663$ & $716$ & $328$ & $194$ & $0.0066475$ \\ \hline $29184$ & $28662$ & $717$ & $328$ & $195$ & $0.0066817$ & $29184$ & $28661$ & $718$ & $328$ & $196$ & $0.0067160$ \\ \hline $29184$ & $28660$ & $719$ & $329$ & $196$ & $0.0067160$ & $29184$ & $28659$ & $720$ & $330$ & $196$ & $0.0067160$ \\ \hline $29184$ & $28658$ & $721$ & $332$ & $195$ & $0.0066817$ & $29184$ & $28657$ & $722$ & $332$ & $196$ & $0.0067160$ \\ \hline $29184$ & $28656$ & $723$ & $333$ & $196$ & $0.0067160$ & $29184$ & $28655$ & $724$ & $334$ & $196$ & $0.0067160$ \\ \hline $29184$ & $28654$ & $725$ & $335$ & $196$ & $0.0067160$ & $29184$ & $28653$ & $726$ & $336$ & $196$ & $0.0067160$ \\ \hline $29184$ & $28652$ & $727$ & $337$ & $196$ & $0.0067160$ & $29184$ & $28651$ & $728$ & $338$ & $196$ & $0.0067160$ \\ \hline $29184$ & $28650$ & $729$ & $339$ & $196$ & $0.0067160$ & $29184$ & $28649$ & $730$ & $340$ & $196$ & $0.0067160$ \\ \hline $29184$ & $28648$ & $731$ & $342$ & $195$ & $0.0066817$ & $29184$ & $28647$ & $732$ & $342$ & $196$ & $0.0067160$ \\ \hline $29184$ & $28646$ & $733$ & $343$ & $196$ & $0.0067160$ & $29184$ & $28645$ & $734$ & $344$ & $196$ & $0.0067160$ \\ \hline $29184$ & $28644$ & $735$ & $345$ & $196$ & $0.0067160$ & $29184$ & $28643$ & $736$ & $346$ & $196$ & $0.0067160$ \\ \hline $29184$ & $28642$ & $737$ & $347$ & $196$ & $0.0067160$ & $29184$ & $28641$ & $738$ & $348$ & $196$ & $0.0067160$ \\ \hline
\end{tabular}\end{scriptsize}\end{center}

\begin{center}\begin{scriptsize}\begin{tabular}{|c|c|c|c|c|c||c|c|c|c|c|c|} 
\hline $n$&$k$&$\rho_{\ell}$&$d_{ORD}$&$\delta \leq $&$\Delta \leq$ & $n$&$k$&$\rho_{\ell}$&$d_{ORD}$&$\delta \leq $&$\Delta \leq$ \\ \hline
 $29184$ & $28640$ & $739$ & $349$ & $196$ & $0.0067160$ & $29184$ & $28639$ & $740$ & $350$ & $196$ & $0.0067160$ \\ \hline $29184$ & $28638$ & $741$ & $352$ & $195$ & $0.0066817$ & $29184$ & $28637$ & $742$ & $352$ & $196$ & $0.0067160$ \\ \hline $29184$ & $28636$ & $743$ & $353$ & $196$ & $0.0067160$ & $29184$ & $28635$ & $744$ & $354$ & $196$ & $0.0067160$ \\ \hline $29184$ & $28634$ & $745$ & $355$ & $196$ & $0.0067160$ & $29184$ & $28633$ & $746$ & $356$ & $196$ & $0.0067160$ \\ \hline $29184$ & $28632$ & $747$ & $357$ & $196$ & $0.0067160$ & $29184$ & $28631$ & $748$ & $358$ & $196$ & $0.0067160$ \\ \hline $29184$ & $28630$ & $749$ & $359$ & $196$ & $0.0067160$ & $29184$ & $28629$ & $750$ & $360$ & $196$ & $0.0067160$ \\ \hline $29184$ & $28628$ & $751$ & $361$ & $196$ & $0.0067160$ & $29184$ & $28627$ & $752$ & $362$ & $196$ & $0.0067160$ \\ \hline $29184$ & $28626$ & $753$ & $363$ & $196$ & $0.0067160$ & $29184$ & $28625$ & $754$ & $364$ & $196$ & $0.0067160$ \\ \hline $29184$ & $28624$ & $755$ & $365$ & $196$ & $0.0067160$ & $29184$ & $28623$ & $756$ & $366$ & $196$ & $0.0067160$ \\ \hline $29184$ & $28622$ & $757$ & $367$ & $196$ & $0.0067160$ & $29184$ & $28621$ & $758$ & $368$ & $196$ & $0.0067160$ \\ \hline $29184$ & $28620$ & $759$ & $369$ & $196$ & $0.0067160$ & $29184$ & $28619$ & $760$ & $370$ & $196$ & $0.0067160$ \\ \hline $29184$ & $28618$ & $761$ & $371$ & $196$ & $0.0067160$ & $29184$ & $28617$ & $762$ & $372$ & $196$ & $0.0067160$ \\ \hline $29184$ & $28616$ & $763$ & $373$ & $196$ & $0.0067160$ & $29184$ & $28615$ & $764$ & $374$ & $196$ & $0.0067160$ \\ \hline $29184$ & $28614$ & $765$ & $375$ & $196$ & $0.0067160$ & $29184$ & $28613$ & $766$ & $376$ & $196$ & $0.0067160$ \\ \hline $29184$ & $28612$ & $767$ & $377$ & $196$ & $0.0067160$ & $29184$ & $28611$ & $768$ & $378$ & $196$ & $0.0067160$ \\ \hline $29184$ & $28610$ & $769$ & $379$ & $196$ & $0.0067160$ & $29184$ & $28609$ & $770$ & $380$ & $196$ & $0.0067160$ \\ \hline $29184$ & $28608$ & $771$ & $381$ & $196$ & $0.0067160$ & $29184$ & $28607$ & $772$ & $382$ & $196$ & $0.0067160$ \\ \hline $29184$ & $28606$ & $773$ & $383$ & $196$ & $0.0067160$ & $29184$ & $28605$ & $774$ & $384$ & $196$ & $0.0067160$ \\ \hline $29184$ & $28604$ & $775$ & $385$ & $196$ & $0.0067160$ & $29184$ & $28603$ & $776$ & $386$ & $196$ & $0.0067160$ \\ \hline $29184$ & $28602$ & $777$ & $387$ & $196$ & $0.0067160$ & $29184$ & $28601$ & $778$ & $388$ & $196$ & $0.0067160$ \\ \hline $29184$ & $28600$ & $779$ & $389$ & $196$ & $0.0067160$ & $29184$ & $28599$ & $780$ & $390$ & $196$ & $0.0067160$ \\ \hline $29184$ & $28598$ & $781$ & $392$ & $195$ & $0.0066817$ & $29184$ & $28597$ & $782$ & $392$ & $196$ & $0.0067160$ \\ \hline $29184$ & $28596$ & $783$ & $393$ & $196$ & $0.0067160$ & $29184$ & $28595$ & $784$ & $394$ & $196$ & $0.0067160$ \\ \hline
\end{tabular}\end{scriptsize}\end{center}

\section{Quantum codes from $\tilde{\mathcal{S}}_q$}\label{SectionCSS}

In this section we construct quantum codes from families of one-point AG codes on $\tcSq$, using the so-called \emph{CSS construction} which allows to construct quantum codes from classical linear codes; see \cite[Lemma 2.5]{LGP}.

Let $q$ be a prime power and $\mathbb{H}=(\mathbb{C}^q)^{\otimes n}=\mathbb{C}^q \otimes \cdots \otimes \mathbb{C}^q$ be a $q^n$-dimensional Hilbert space. Then a $q$-ary quantum code $C$ of length $n$ and dimension $k$ is defined as a $q^k$-dimensional Hilbert subspace of $\mathbb{H}$.
If $C$ has minimum distance $d$, then $C$ can correct up to $\lfloor \frac{d-1}{2}\rfloor$ errors. Such a quantum code is denoted by $[[n,k,d]]_q$.
The quantum Singleton bound $2d+k\leq 2+n$ holds for a $[[n,k,d]]_q$-quantum code.
The quantum Singleton defect and the relative quantum Singleton defect are defined by $\delta_Q:= n-k-2d+2$ and $\Delta_Q:=\frac{\delta_Q}{n}$, respectively. 
If $\delta_Q=0$, then the code is said to be \emph{quantum MDS}.

\begin{lemma}{\rm{\cite[Lemma 6.1]{GMB2}}}{\rm{(CSS construction)}}\label{CSS}
Let $C_1$ and $C_2$ be linear codes with parameters $[n,k_1,d_1]_q$ and $[n,k_2,d_2]_q$, respectively, and assume that $C_1 \subset C_2$. Then there exists a $[[n,k_2-k_1,d]]_q$-quantum code with $$d=\min \{ w(c)\vert c\in (C_2 \setminus C_1)\cup (C_1^{\perp} \setminus C_2^{\perp}) \}.$$
\end{lemma}

As an application of Lemma \ref{CSS} to AG codes, La Guardia and Pereira proposed in \cite{LGP} the following \emph{general $t$-point construction}.
\begin{lemma}{\rm{\cite[Theorem 3.1]{LGP}}}{\rm{(general $t$-point construction)}}\label{generaltpoint}
Let $\mathcal{X}$ be a non-singular curve on $\mathbb{F}_q$ of genus $g$ and with $n+t$ distinct points $\mathbb{F}_q$-rational for some $n,t>0$. For every $i=1,...,t$, let $a_i,b_i$ be positive integers such that $a_i\leq b_i$ and $$2g-2<\sum_{i=1}^t a_i<\sum_{i=1}^t b_i < n .$$ 
Then there exists a $[[n,k,d]]_q$-quantum code with $k=\sum_{i=1}^t b_i-\sum_{i=1}^t a_i$ and \\$d\geq \min\lbrace n-\sum_{i=1}^t b_i, \sum_{i=1}^t a_i-(2g-2)\rbrace$.
\end{lemma}

By applying Lemma \ref{generaltpoint} to one-point AG codes on $\tilde{\mathcal{S}}_q$, the following result is obtained.

\begin{proposition}
Let $a,b\in \mathbb{N}$ such that $$q^3-2q^2+q-2<a<b<q^5-q^4+q^3.$$
Then there exists a quantum code with parameters $[[n,b-a,d]]_{q^4}$, where $$n=q^5-q^4+q^3,$$ $$d\geq \min\lbrace q^5-q^4+q^3-b,\; a-(q^3-2q^2+q-2) \rbrace .$$
\end{proposition}

We can also apply the CSS construction to dual codes of one-point AG codes on $\tilde{S}_q$.
We keep the notation of Section \ref{codiciduali}.
Let $a=\rho_\ell \in H(P_{\infty})$ and $b=\rho_{\ell+s}\in H(P_{\infty})$ with $s\geq 1$. Denote by $C_1$ and $C_2$ the codes $C_{\ell+s}:=C_{\ell+s}(P_{\infty})$ and $C_\ell:=C_\ell(P_{\infty})$, respectively. Then $C_1\subset C_2$. Also,
the dimensions $k_1$ and $k_2$ respectively of $C_1$ and $C_2$ satisfy $k_2=n-h_{\ell}$ and $k_1=n-h_{\ell+s}=n-h_\ell-s$, where $h_i$ is the number of non-gaps at $P_{\infty}$ which are smaller than or equal to $i$.
The CSS construction now provides a $[[n,s,d]]_{q^4}$-quantum code with $n=q^5-q^4+q^3$ and $d =\min \lbrace w(c) \vert c \in (C_\ell \setminus C_{\ell+s})\cup (C(D,G_1)\setminus C(D,G_2))) \rbrace $, where $G_1=\rho_{\ell+s}P_{\infty}$ and $G_2=\rho_\ell P_{\infty}$.
Hence,
\begin{equation}
\label{dordmin}
d\geq \min\lbrace d_{ORD}(C_\ell),d_1 \rbrace,
\end{equation}
where $d_1$ is the minimum distance of $C(D,G_1)$.
Choose $\ell \in [3g-1,n-g]$ and $s\in [1,n-2\ell]$. As $\ell\geq3g-1$, $\rho_{\ell+s}=g-1+\ell+s$; thus $d_1\geq n-\deg(G_1)=n-\rho_{\ell+s}\geq n-\ell -s-g+1$. Also, from Theorem \ref{Ssimmetrico} and Proposition \ref{dord1}, $d_{ORD}(C_\ell)=\ell+1-g$. Since $s\leq n-2\ell$, Equation \eqref{dordmin} reads $d\geq d_{ORD}(C_\ell)=\ell+1-g$.
Hence, the following result has been proved.

\begin{theorem}\label{quant1}
Let $n=q^5-q^4+q^3$. Then for every $\ell \in [3g-1,n-g]$ and $s\in [1,n-2\ell]$ there exists a $[[n,s,d]]_{q^4}$-quantum code with $d\geq \ell+1-g$.
\end{theorem}

Choose $s=n-2\ell$. Then Theorem \ref{quant1} provides a $[[n,s,d]]_{q^4}$-quantum code whose relative quantum Singleton defect $\Delta_Q$ is upper bounded as follows:
$$\Delta_Q=\frac{n-s-2d+2}{n}=\frac{2\ell-2d+2}{n}\leq \frac{2g}{n}=\frac{q^3-2q^2+q}{q^5-q^4+q^3}.$$

The following table show the parameters of such $[[n,s,d]]_{q^4}$-quantum codes obtained through the CSS construction in the case $q=8$.
Here, $s=n-2\ell$ with $\ell$ ranging in $2g,\ldots,3g-2$, that is when Theorem \ref{quant1} does not apply.
The minimum distance $d_1$ is lower bounded by the Goppa minimum distance $d_1^*=n-\deg(G_1)=n-\rho_{\ell+s}=\ell+1-g$; hence, Equation \eqref{dordmin} reads $d\geq\ell+1-g$ which is the bound used in the table. Also, $\Delta_Q$ is upper bounded by $\frac{n-s-2d_1^*+2}{n}$.

\begin{center}
\begin{scriptsize}
\begin{tabular}{|c|c|c|c||c|c|c|c||c|c|c|c|}
\hline
$n$&$s$&$d\geq$&$\Delta_Q \leq$&$n$&$s$&$d\geq$&$\Delta_Q \leq$&$n$&$s$&$d\geq$&$\Delta_Q \leq$ \\ \hline 

  $29184$ & $28400$ & $197$ & $0.013432$ & $29184$ & $28398$ & $198$ & $0.013432$ & $29184$ & $28396$ & $199$ & $0.013432$ \\ \hline $29184$ & $28394$ & $200$ & $0.013432$ & $29184$ & $28392$ & $201$ & $0.013432$ & $29184$ & $28390$ & $202$ & $0.013432$ \\ \hline $29184$ & $28388$ & $203$ & $0.013432$ & $29184$ & $28386$ & $204$ & $0.013432$ & $29184$ & $28384$ & $205$ & $0.013432$ \\ \hline $29184$ & $28382$ & $206$ & $0.013432$ & $29184$ & $28380$ & $207$ & $0.013432$& $29184$ & $28378$ & $208$ & $0.013432$ \\ \hline $29184$ & $28376$ & $209$ & $0.013432$ & $29184$ & $28374$ & $210$ & $0.013432$& $29184$ & $28372$ & $211$ & $0.013432$ \\ \hline $29184$ & $28370$ & $212$ & $0.013432$ & $29184$ & $28368$ & $213$ & $0.013432$ & $29184$ & $28366$ & $214$ & $0.013432$ \\ \hline $29184$ & $28364$ & $215$ & $0.013432$ & $29184$ & $28362$ & $216$ & $0.013432$ & $29184$ & $28360$ & $217$ & $0.013432$ \\ \hline $29184$ & $28358$ & $218$ & $0.013432$ & $29184$ & $28356$ & $219$ & $0.013432$ & $29184$ & $28354$ & $220$ & $0.013432$ \\ \hline $29184$ & $28352$ & $221$ & $0.013432$ & $29184$ & $28350$ & $222$ & $0.013432$ & $29184$ & $28348$ & $223$ & $0.013432$ \\ \hline $29184$ & $28346$ & $224$ & $0.013432$ & $29184$ & $28344$ & $225$ & $0.013432$ & $29184$ & $28342$ & $226$ & $0.013432$ \\ \hline $29184$ & $28340$ & $227$ & $0.013432$ & $29184$ & $28338$ & $228$ & $0.013432$ & $29184$ & $28336$ & $229$ & $0.013432$ \\ \hline $29184$ & $28334$ & $230$ & $0.013432$ & $29184$ & $28332$ & $231$ & $0.013432$ & $29184$ & $28330$ & $232$ & $0.013432$ \\ \hline $29184$ & $28328$ & $233$ & $0.013432$ & $29184$ & $28326$ & $234$ & $0.013432$ & $29184$ & $28324$ & $235$ & $0.013432$ \\ \hline $29184$ & $28322$ & $236$ & $0.013432$ & $29184$ & $28320$ & $237$ & $0.013432$ & $29184$ & $28318$ & $238$ & $0.013432$ \\ \hline $29184$ & $28316$ & $239$ & $0.013432$ & $29184$ & $28314$ & $240$ & $0.013432$ & $29184$ & $28312$ & $241$ & $0.013432$ \\ \hline $29184$ & $28310$ & $242$ & $0.013432$ & $29184$ & $28308$ & $243$ & $0.013432$ & $29184$ & $28306$ & $244$ & $0.013432$ \\ \hline $29184$ & $28304$ & $245$ & $0.013432$ & $29184$ & $28302$ & $246$ & $0.013432$ & $29184$ & $28300$ & $247$ & $0.013432$ \\ \hline $29184$ & $28298$ & $248$ & $0.013432$ & $29184$ & $28296$ & $249$ & $0.013432$ & $29184$ & $28294$ & $250$ & $0.013432$ \\ \hline $29184$ & $28292$ & $251$ & $0.013432$ & $29184$ & $28290$ & $252$ & $0.013432$ & $29184$ & $28288$ & $253$ & $0.013432$ \\ \hline $29184$ & $28286$ & $254$ & $0.013432$ & $29184$ & $28284$ & $255$ & $0.013432$ & $29184$ & $28282$ & $256$ & $0.013432$ \\ \hline $29184$ & $28280$ & $257$ & $0.013432$ & $29184$ & $28278$ & $258$ & $0.013432$ & $29184$ & $28276$ & $259$ & $0.013432$ \\ \hline $29184$ & $28274$ & $260$ & $0.013432$ & $29184$ & $28272$ & $261$ & $0.013432$ & $29184$ & $28270$ & $262$ & $0.013432$ \\ \hline $29184$ & $28268$ & $263$ & $0.013432$ & $29184$ & $28266$ & $264$ & $0.013432$ & $29184$ & $28264$ & $265$ & $0.013432$ \\ \hline $29184$ & $28262$ & $266$ & $0.013432$ & $29184$ & $28260$ & $267$ & $0.013432$ & $29184$ & $28258$ & $268$ & $0.013432$ \\ \hline $29184$ & $28256$ & $269$ & $0.013432$ & $29184$ & $28254$ & $270$ & $0.013432$ & $29184$ & $28252$ & $271$ & $0.013432$ \\ \hline $29184$ & $28250$ & $272$ & $0.013432$ & $29184$ & $28248$ & $273$ & $0.013432$ & $29184$ & $28246$ & $274$ & $0.013432$ \\ \hline $29184$ & $28244$ & $275$ & $0.013432$ & $29184$ & $28242$ & $276$ & $0.013432$ & $29184$ & $28240$ & $277$ & $0.013432$ \\ \hline $29184$ & $28238$ & $278$ & $0.013432$ & $29184$ & $28236$ & $279$ & $0.013432$ & $29184$ & $28234$ & $280$ & $0.013432$ \\ \hline $29184$ & $28232$ & $281$ & $0.013432$ & $29184$ & $28230$ & $282$ & $0.013432$ & $29184$ & $28228$ & $283$ & $0.013432$ \\ \hline $29184$ & $28226$ & $284$ & $0.013432$ & $29184$ & $28224$ & $285$ & $0.013432$& $29184$ & $28222$ & $286$ & $0.013432$ \\ \hline \end{tabular}\end{scriptsize}\end{center} 
\begin{center}
\begin{scriptsize}
\begin{tabular}{|c|c|c|c||c|c|c|c||c|c|c|c|}
\hline
$n$&$s$&$d\geq$&$\Delta_Q \leq$&$n$&$s$&$d\geq$&$\Delta_Q \leq$&$n$&$s$&$d\geq$&$\Delta_Q \leq$ \\ \hline    $29184$ & $28220$ & $287$ & $0.013432$ & $29184$ & $28218$ & $288$ & $0.013432$ & $29184$ & $28216$ & $289$ & $0.013432$ \\ \hline $29184$ & $28214$ & $290$ & $0.013432$ & $29184$ & $28212$ & $291$ & $0.013432$ & $29184$ & $28210$ & $292$ & $0.013432$ \\ \hline $29184$ & $28208$ & $293$ & $0.013432$ & $29184$ & $28206$ & $294$ & $0.013432$ & $29184$ & $28204$ & $295$ & $0.013432$ \\ \hline $29184$ & $28202$ & $296$ & $0.013432$& $29184$ & $28200$ & $297$ & $0.013432$ & $29184$ & $28198$ & $298$ & $0.013432$ \\ \hline $29184$ & $28196$ & $299$ & $0.013432$ & $29184$ & $28194$ & $300$ & $0.013432$ & $29184$ & $28192$ & $301$ & $0.013432$ \\ \hline $29184$ & $28190$ & $302$ & $0.013432$ & $29184$ & $28188$ & $303$ & $0.013432$ & $29184$ & $28186$ & $304$ & $0.013432$ \\ \hline $29184$ & $28184$ & $305$ & $0.013432$ & $29184$ & $28182$ & $306$ & $0.013432$ & $29184$ & $28180$ & $307$ & $0.013432$ \\ \hline $29184$ & $28178$ & $308$ & $0.013432$ & $29184$ & $28176$ & $309$ & $0.013432$ & $29184$ & $28174$ & $310$ & $0.013432$ \\ \hline $29184$ & $28172$ & $311$ & $0.013432$ & $29184$ & $28170$ & $312$ & $0.013432$ & $29184$ & $28168$ & $313$ & $0.013432$ \\ \hline $29184$ & $28166$ & $314$ & $0.013432$ & $29184$ & $28164$ & $315$ & $0.013432$ & $29184$ & $28162$ & $316$ & $0.013432$ \\ \hline $29184$ & $28160$ & $317$ & $0.013432$ & $29184$ & $28158$ & $318$ & $0.013432$ & $29184$ & $28156$ & $319$ & $0.013432$ \\ \hline $29184$ & $28154$ & $320$ & $0.013432$ & $29184$ & $28152$ & $321$ & $0.013432$ & $29184$ & $28150$ & $322$ & $0.013432$ \\ \hline $29184$ & $28148$ & $323$ & $0.013432$ & $29184$ & $28146$ & $324$ & $0.013432$ & $29184$ & $28144$ & $325$ & $0.013432$ \\ \hline $29184$ & $28142$ & $326$ & $0.013432$ & $29184$ & $28140$ & $327$ & $0.013432$ & $29184$ & $28138$ & $328$ & $0.013432$ \\ \hline $29184$ & $28136$ & $329$ & $0.013432$ & $29184$ & $28134$ & $330$ & $0.013432$ & $29184$ & $28132$ & $331$ & $0.013432$ \\ \hline $29184$ & $28130$ & $332$ & $0.013432$ & $29184$ & $28128$ & $333$ & $0.013432$ & $29184$ & $28126$ & $334$ & $0.013432$ \\ \hline $29184$ & $28124$ & $335$ & $0.013432$ & $29184$ & $28122$ & $336$ & $0.013432$ & $29184$ & $28120$ & $337$ & $0.013432$ \\ \hline $29184$ & $28118$ & $338$ & $0.013432$ & $29184$ & $28116$ & $339$ & $0.013432$ & $29184$ & $28114$ & $340$ & $0.013432$ \\ \hline $29184$ & $28112$ & $341$ & $0.013432$ & $29184$ & $28110$ & $342$ & $0.013432$ & $29184$ & $28108$ & $343$ & $0.013432$ \\ \hline $29184$ & $28106$ & $344$ & $0.013432$ & $29184$ & $28104$ & $345$ & $0.013432$ & $29184$ & $28102$ & $346$ & $0.013432$ \\ \hline $29184$ & $28100$ & $347$ & $0.013432$ & $29184$ & $28098$ & $348$ & $0.013432$ & $29184$ & $28096$ & $349$ & $0.013432$ \\ \hline $29184$ & $28094$ & $350$ & $0.013432$ & $29184$ & $28092$ & $351$ & $0.013432$ & $29184$ & $28090$ & $352$ & $0.013432$ \\ \hline $29184$ & $28088$ & $353$ & $0.013432$ & $29184$ & $28086$ & $354$ & $0.013432$ & $29184$ & $28084$ & $355$ & $0.013432$ \\ \hline $29184$ & $28082$ & $356$ & $0.013432$ & $29184$ & $28080$ & $357$ & $0.013432$ & $29184$ & $28078$ & $358$ & $0.013432$ \\ \hline $29184$ & $28076$ & $359$ & $0.013432$ & $29184$ & $28074$ & $360$ & $0.013432$ & $29184$ & $28072$ & $361$ & $0.013432$ \\ \hline $29184$ & $28070$ & $362$ & $0.013432$ & $29184$ & $28068$ & $363$ & $0.013432$ & $29184$ & $28066$ & $364$ & $0.013432$ \\ \hline $29184$ & $28064$ & $365$ & $0.013432$ & $29184$ & $28062$ & $366$ & $0.013432$ & $29184$ & $28060$ & $367$ & $0.013432$ \\ \hline $29184$ & $28058$ & $368$ & $0.013432$ & $29184$ & $28056$ & $369$ & $0.013432$ & $29184$ & $28054$ & $370$ & $0.013432$ \\ \hline $29184$ & $28052$ & $371$ & $0.013432$ & $29184$ & $28050$ & $372$ & $0.013432$ & $29184$ & $28048$ & $373$ & $0.013432$ \\ \hline $29184$ & $28046$ & $374$ & $0.013432$ & $29184$ & $28044$ & $375$ & $0.013432$ & $29184$ & $28042$ & $376$ & $0.013432$ \\ \hline $29184$ & $28040$ & $377$ & $0.013432$ & $29184$ & $28038$ & $378$ & $0.013432$ & $29184$ & $28036$ & $379$ & $0.013432$ \\ \hline $29184$ & $28034$ & $380$ & $0.013432$ & $29184$ & $28032$ & $381$ & $0.013432$ & $29184$ & $28030$ & $382$ & $0.013432$ \\ \hline $29184$ & $28028$ & $383$ & $0.013432$ & $29184$ & $28026$ & $384$ & $0.013432$& $29184$ & $28024$ & $385$ & $0.013432$ \\ \hline $29184$ & $28022$ & $386$ & $0.013432$ & $29184$ & $28020$ & $387$ & $0.013432$ & $29184$ & $28018$ & $388$ & $0.013432$ \\ \hline $29184$ & $28016$ & $389$ & $0.013432$ & $29184$ & $28014$ & $390$ & $0.013432$ & $29184$ & $28012$ & $391$ & $0.013432$ \\ \hline $29184$ & $28010$ & $392$ & $0.013432$ & $29184$ & $28008$ & $393$ & $0.013432$ & $29184$ & $28006$ & $394$ & $0.013432$ \\ \hline
\end{tabular}
\end{scriptsize}
\end{center}

\section{Weak Castle codes from $\tilde{\mathcal{S}}_q$} \label{castle}
In this section we construct weak Castle codes from the curve $\tilde{\mathcal{S}}_q$; we refer to \cite{Castle} for notation and results on weak Castle curves and weak Castle codes.

Let $\mathcal{X}$ be a curve on $\mathbb{F}_q$ and $Q$ be an $\mathbb F_q$-rational place of $\mathcal X$. Then the pair $(\mathcal{X},Q)$ is called \textit{weak Castle} if the following conditions are satisfied.
\begin{itemize}
\item[$C1)$] The Weierstrass semigroup $H(Q)$ is symmetric.
\item[$C2)$] For some integer $\ell$, there exists a morphism $f : \mathcal{X}\rightarrow \mathbb{P}^1 = \overline{\mathbb F}_q\cup\{\infty\}$ such that $(f)_{\infty}=\ell Q$ and there exists a set $U=\lbrace \alpha_1,...,\alpha_h \rbrace \subseteq \mathbb{F}_q$, such that for every $i=1,...,h$, $ f^{-1}(\alpha_i)\subseteq \mathcal{X}(\mathbb{F}_q)$ and $\vert f^{-1}(\alpha_i) \vert =\ell$.
\end{itemize}
If $(\mathcal{X},Q)$ is weak Castle, define
$$D=\sum_{i=1}^h \sum_{j=1}^\ell P_j^i,$$
where $f^{-1}(\alpha_i)=\lbrace P_1^i,\ldots,P_\ell ^i\rbrace$ for every $i=1,\ldots,h$.
The one-point AG codes $C(D,rQ)$ are called \emph{weak Castle codes}.

\begin{proposition}{\rm{(\!\!\cite[Proposition 1, Lemma 2, and Corollary 2]{Fernandello})}}\label{WeakProp}
Let $(\mathcal{X},Q)$ be a weak Castle curve of genus $g$ and $C(D,rQ)$ be a weak Castle code from $\mathcal{X}$.
Define $r^{\perp}=n+2g-2-r$.
Then the following properties hold:
\begin{itemize}
\item[(i)] Under the notation in Condition {\rm C2)}, let $\varphi=\prod_{i=1}^h(f-\alpha_i)$. If $div(d\varphi)=(2g-2)Q$, then $C(D,rQ)^\perp = C(D,r^\perp Q)$.
\item[(ii)] The divisors $D$ and $rQ$ are equivalent. Also, for every $r<n$, $C(D,rQ)$ attains the designed minimum distance $d^*$ if and only if $C(D,(n-r)Q)$ attains the designed minimum distance as well.
\item[(iii)] $(2g-2)Q$ and $(n+2g-2)Q-D$ are canonical divisors, and there exists $x\in (\mathbb{F}_{q}^*)^n$ such that $C(D,rQ)^{\perp}=x\cdot C(D, r^{\perp}Q)$.
\item[(iv)] For every $i=1,...,r$, let $r_i:=\min \lbrace r: \ell (rQ)-\ell ((r-n)Q)\geq i \rbrace$ and $C_i:=C(D,r_iQ)$.
Then $C_i$ has dimension $i$, and
$$C_0=(0)\subset C_1\subset \cdots \subset C_n=\mathbb{F}_q^n$$ is a formally self-dual sequence of codes.
\item[(v)] If $2i\leq n$, then there exist quantum codes with parameters $[[n,n-2i,\geq d(C_{n-i})]]_{q}$ where $d(C_{n-i})\geq n-r_{n-i}+\gamma_{a+1}$, with $a=\ell((r_{n-i}-n)Q)$ and $$\gamma_{a+1}=\min \lbrace \deg(A) : A \textrm{ is a rational divisor on } \mathcal{X} \textrm{ with } \ell(A)\geq a+1\rbrace.$$
\end{itemize}
\end{proposition}

We construct weak Castle codes from $\tcSq$.

\begin{proposition}\label{sqcastle}
The pair $(\tilde{\mathcal{S}}_q, P_{\infty})$ is weak Castle.
\end{proposition}
\begin{proof}
By Theorem \ref{Ssimmetrico}, Condition C1) is satisfied.
From \cite[Proposition 25]{GMQZ},
$$ T:=\{(x,y,t)\mapsto(x+b,y+b^{q_0}x+c,t)\mid b,c\in\mathbb{F}_q\} $$
is a Sylow $2$-subgroup of ${\rm Aut}(\tcSq)$. As $S$ fixes $t$, we have $\mathbb{K}(t)\subseteq\mathbb{K}(\tcSq/T)$.
On the other side, $\tcSq/T$ is rational by \cite[Theorem 11.78]{HKT}.
Hence, $\mathbb{K}(t)=\mathbb{K}(\tcSq/T)$.
The pole divisor of $t$ is $(t)_\infty = q^2 P_\infty$.
Also, $T$ acts semiregularly on $\tcSq\setminus\{P_\infty\}$; being $\mathbb{F}_{q^4}$-rational, $T$ acts semiregularly on the $q^5-q^4+q^3$ places in $\tcSq(\mathbb{F}_{q^4})\setminus\{P_\infty\}$. Let $U=\{\alpha_1,\ldots,\alpha_{q^3-q^2+q}\}\subseteq\mathbb{F}_{q^4}$ be the set of places of $\mathbb{P}^1\cong\tcSq/T$ lying under $\tcSq(\mathbb{F}_{q^4})\setminus\{P_\infty\}$. Then Condition C2) is satisfied by choosing $\ell=q^2$, $f=t$, $h={q^3-q^2+q}$.
\end{proof}

From Proposition \ref{sqcastle}, we have weak Castle codes $C(D,rP_{\infty})$ of length $n=q^5-q^4+q^3$ with $D=\sum_{P\in\tcSq(\mathbb{F}_{q^4})\setminus\{P_\infty\}} P$.

Now we construct quantum codes by means of the CSS construction.
Let $\varphi$ be defined as in Proposition \ref{WeakProp} {\it (i)}. By direct computation, $P_\infty$ is the only pole of $d\varphi$ with multiplicity $2g-2$, and the zeros of $\varphi$ are simple; hence, $div(d\phi)=(2g-2)P_{\infty}$. Thus, by Proposition \ref{WeakProp}, $C(D,rP_\infty)^\perp=C(D,r^\perp P_\infty)$.
Let $\rho_a, \rho_{a+b}\in H(P_{\infty})\cap [2,n-2]$, with $a,b\geq 1$ and consider the codes $C_{a+b}:=C^{\perp}(D,\rho_{a+b}P_{\infty})$ and $C_a:=C^{\perp}(D,\rho_{a}P_{\infty})$, which satisfy $C_{a+b} \subseteq C_a$. Then the CSS construction yields a $[[n,s,d]]_{q^4}$-quantum code such that
$d\geq \min \lbrace d_{ORD}(C_a),d_{1} \rbrace$, where $d_1$ is the minimum distance of the cose $C(D,\rho_{a+b}P_{\infty})$. Since $C(D,\rho_{a+b}P_{\infty})=C^{\perp}(D,\rho_{a+b}^{\perp}P_{\infty})$, the lower bound on $d$ reads
\begin{equation}\label{stima}
d\geq\min\lbrace d_{ORD}(C_a), d_{ORD}(C^{\perp}(D,\rho_{a+b}^{\perp}P_{\infty})) \rbrace.
 \end{equation}
As the tables in Section \ref{codiciduali} show, the Feng-Rao minimum distance often improves the Goppa designed minimum distance. Thus, Bound \eqref{stima} improves in general the lower bound $d\geq\min\{d_{ORD}(C_a),d_1\}$.


\section{Convolutional codes from $\tilde{\mathcal{S}}_q$} \label{convo}

In this section we construct convolutional codes using a result by  De Assis, La Guardia, and Pereira \cite{ALGR}which allows to construct unit-memory convolutional codes with certain parameters $(n,k,\gamma;m,d_f)$ starting from AG codes.
Consider the polynomial ring $R=\mathbb{F}_q[X]$. A convolutional code $C$ is an $R$-submodule of rank $k$ of the module $R^n$. Let $G(X)=(g_{ij}(X))\in \mathbb{F}_q[X]^{k\times n}$ be a generator matrix of $C$ over $\mathbb{F}_q[X]$, $\gamma_i = \max\lbrace \deg g_{ij}(X) \mid 1\leq j\leq n\rbrace $, $\gamma=\sum_{i=1}^k \gamma_i$, $m=\max \lbrace \gamma_i \mid 1\leq i\leq k\rbrace $, and $d_f$ be the minimum weight of a word $c\in C$.
Then $C$ is an $(n,k,\gamma ;m,d_f)_q$-convolutional code, having length $n$, dimension $k$, degree $\gamma$, memory $m$, and distance $d_f$.
When $m=1$, $C$ is said to be a unit-memory convolutional code.
For a detailed introduction to convolutional codes, see \cite{ALGR,RS1999}.

\begin{lemma}[\!\!\cite{ALGR}, Theorem 3]\label{Convol}
Let $\mathcal{X}$ be an $\mathbb{F}_q$-rational non-singular curve of genus $g$.
Consider an AG code $C^\perp(D,G)$ with $\mathcal X$ with $2g-2<\deg(G)<n$. Then there exists a unit-memory convolutional code with parameters $(n,k-\ell,\ell;1,d_f\geq d)$, where $\ell \leq k/2$, $k=\deg(G)+1-g$ and $d\geq n-\deg(G)$.
\end{lemma}

We apply Lemma \ref{Convol} to one-point AG codes from $\tilde{\mathcal{S}}_q$.

\begin{proposition}\label{costrconv}
Consider the curve $\tcSq$ and a non-gap $\rho_\ell \in H(P_{\infty})$ at $P_\infty$ such that $q^3-2q^2+q-2<\rho_\ell<q^5-q^4+q^3$. Then there exists a unit-memory convolutional code with parameters $(n,k-s,s;1,d_f\geq d_{ORD}(C_{\ell}(P_{\infty})))_{q^4}$, where $n=q^5-q^4+q^3$, $s\leq k/2$, and $k=\rho_\ell +1-\frac{q^3-2q^2+q}{2}$.
\end{proposition}
\begin{proof}
The claim follows by applying Lemma \ref{Convol} to the code $C_{\ell}(P_{\infty})$ constructed as in Section \ref{codiciduali}, and using that $d_f\geq d(C_{\ell}(P_{\infty}))\geq d_{ORD}(C_{\ell}(P_{\infty}))$.
\end{proof}

\begin{corollary}
Consider the curve $\tcSq$ and a non-gap $\rho_\ell \in H(P_{\infty})$ at $P_\infty$ such that $q^3-2q^2+q-2<\rho_\ell<n$ where $n=q^5-q^4+q^3$. Let $\ell\geq \frac{3(q^3-2q^2+q)}{2}-1$. Then there exists a unit-memory convolutional code with parameters $(n,k-s,s;1,d_f)_{q^4}$, where $k=\rho_\ell +1-\frac{q^3-2q^2+q}{2}$, $s\leq k/2$, and $d_f\geq \ell+1-\frac{q^3-2q^2+q}{2}$.
\end{corollary}
\begin{proof}
Since $\ell \geq 3g-1$, $ d_{ORD}(C_{\ell}(P_{\infty}))=\ell +1-g$ by Proposition \ref{dord1}. The claim follows from Proposition \ref{costrconv}.
\end{proof}

\end{document}